\documentclass[11pt]{article}
\usepackage{booktabs}
\usepackage[margin=1in]{geometry}
\usepackage[english]{babel}
\usepackage[utf8]{inputenc}
\usepackage[T1]{fontenc}
\usepackage{amsmath,amssymb,amsthm}
\usepackage{tikz}
\usepackage{xcolor}
\usepackage{enumitem}
\usepackage[colorlinks=true,urlcolor=blue,citecolor=blue,linkcolor=blue]{hyperref}
\usepackage{graphicx}

\allowdisplaybreaks
\numberwithin{equation}{section}

\newcommand{\I}{\mathbf{1}}
\newcommand{\ii}{\mathrm{i}}
\newcommand{\e}{\mathrm{e}}
\newcommand{\cI}{\mathcal{I}}

\newcommand{\cS}{\mathcal{S}}
\newcommand{\ordo}{\mathcal{O}}
\newcommand{\cZ}{\mathcal{Z}}
\newcommand{\Spec}{\operatorname{Spec}}

\newcommand{\tA}{\widetilde A}
\newcommand{\tB}{\widetilde B}
\newcommand{\hB}{\widehat B}
\newcommand{\hA}{\widehat A}

\newcommand{\eps}{\varepsilon}
\newcommand{\tbeta}{\widetilde\beta}
\newcommand{\tlambda}{\widetilde\lambda}
\newcommand{\talpha}{\widetilde\alpha}

\newtheorem{theorem}{Theorem}[section]
\newtheorem{proposition}[theorem]{Proposition}
\newtheorem{lemma}[theorem]{Lemma}
\newtheorem{corollary}[theorem]{Corollary}
\newtheorem{remark}[theorem]{Remark}
\newtheorem{definition}[theorem]{Definition}
\usepackage{authblk}

\title{Free fermions in disguise without exponential degeneracies}
\author[1]{Bal\'azs Pozsgay}
\affil[1]{MTA-ELTE “Momentum” Integrable Quantum Dynamics Research Group,\protect\\
    ELTE E\"otv\"os Lor\'and University, Budapest, Hungary}
\date{}

\begin{document}
\maketitle

\begin{abstract}
  Recently, a number of spin chain models have been discovered that are solvable via hidden free-fermionic structures,
  going beyond the Jordan-Wigner paradigm. However, all examples in the literature displayed 
  degeneracies that grow exponentially with the volume and that are homogeneous in the spectrum (identical
  degeneracies for all energy levels).
  In this note we present a model that can be solved by ``free fermions in disguise'' (FFD),
  such that the spectrum is free from exponential degeneracies for generic coupling constants. The model can be seen
  as a particular 
  perturbation of two Ising chains. Alternatively, it can be realized as an interpolation between a standard
  Jordan-Wigner solvable chain and the original FFD model of Fendley. We used ChatGPT Pro 5.4 and 5.5 as a research
  assistant; 
  in the Supplemental Material we provide details about the collaboration between the AI and the human author.
\end{abstract}

\section{Introduction}
\label{sec:introduction}

Finding exactly solvable models is of continuing interest in quantum many-body physics, or theoretical physics more
generally. Arguably the simplest models are those that can be solved by free bosons or free fermions. In certain cases
a model Hamiltonian can appear to be interacting; nevertheless it can be brought into a quadratic form in free
variables/operators by using an appropriate transformation.

In quantum spin chains the most famous example for such an operation is the Jordan-Wigner transformation
\cite{jordan-wigner}, which introduces Dirac or Majorana fermions acting on the Hilbert space of the spin-1/2 chain. The
Jordan-Wigner transformation has a canonical form that is used routinely to treat, for example, the quantum Ising chain
or the XY models, but it also has various algebraic generalizations (see, for example,
\cite{japan-JW-gen-1,japan-JW-gen-2,japan-JW-gen-3}). The conditions for generalized Jordan-Wigner solvability were
determined recently in the independent works \cite{japan-free-fermion-JW,chapman-jw}. 

In a seminal work \cite{ffd} Fendley discovered a spin chain model that was shown to be free-fermionic, although its
solution does not fit into the Jordan-Wigner framework. The title of the paper was ``free fermions in disguise'' (FFD), and
this name was subsequently used in several follow-up papers to denote the model or families of such models.
The methods of \cite{ffd} were soon generalized to similar models in
\cite{fermions-behind-the-disguise}, where it was also shown that the FFD model cannot be solved by a generalized
Jordan-Wigner transformation.
In this work we do not attempt to review the literature on the FFD model and its various
generalizations; instead we refer the reader to a recent work \cite{sajat-solving-ffd-2} where an in-depth review is given.

We instead
focus on a particular problem that has so far plagued all FFD-type models: the exponential degeneracies in the spectrum.
It was already recognized in \cite{ffd} that in the FFD model the number of free-fermionic eigenmodes is far less than
the number of qubits on which the model is defined. As a result, every energy level is degenerate, so that the
degeneracy is homogeneous along the spectrum. The number of fermionic eigenmodes typically grows linearly with the
length of the 
spin chain, but with a coefficient smaller than 1. As a result, there will be an extensive number of 
zero modes that generate a very non-trivial symmetry algebra \cite{ffd,eric-lorenzo-ffd-1}. These zero modes are
responsible for the 
exponential degeneracies.

Based on the seminal works it appeared that such exponential degeneracies are inevitable. The technical reason
for the missing 
eigenmodes can be formulated using algebra and graph theory \cite{fermions-behind-the-disguise} (see also
\cite{sajat-solving-ffd-1} and Subsection \ref{sec:origin} below).
It is notable that exponential degeneracies were present even
in those models \cite{sajat-FP-model,sajat-claws} whose construction evaded some of the limitations presented in
\cite{fermions-behind-the-disguise}.
However, no compelling physical argument was known as to why it has to be this way.

In this work we present an FFD-type model that is free from exponential degeneracies (for generic coupling
constants). In other words, the number of fermionic eigenmodes in the model is either equal to the number of qubits in the
chain, or the difference is only $\ordo(1)$. We present the model in an abstract algebraic formulation, but we also
provide two representations. The first one is via two coupled quantum Ising chains. The second representation can be
seen as an interpolation between the XY chain (given by the Dzyaloshinskii–Moriya interaction term) and the original
FFD model of Fendley. 

This interpolation is one of the key results of our work. It shows that the FFD-type models and the Jordan-Wigner
solvable ones are not truly separate families. Instead, the Jordan-Wigner solvable models are simply very special
examples in the wide family of free-fermionic models. This point of view was already advocated in
\cite{unified-graph-th} and more recently in \cite{sajat-solving-ffd-1,sajat-solving-ffd-2}. However, to our knowledge
the present model 
is the first one that allows for a continuous transition between these apparently different families.

Our final Hamiltonians are linear combinations of two Hamiltonians, each of which is free-fermionic, and their
algebraic structures are compatible with each other. A single Hamiltonian only has half the number of eigenmodes needed,
but the linear combination of the two Hamiltonians finally leads to a spectrum that is generally non-degenerate or
has a homogeneous degeneracy of 2 or 4. The mechanism for the mutual commutativity of the algebraic structures is very similar to
the mechanism that underlies the free-fermionic solvability of the model of \cite{sajat-FP-model}. A crucial difference is
that in the model of \cite{sajat-FP-model} there was just one Hamiltonian (with half the number of eigenmodes needed) and
there was no obvious way to introduce a second, complementary Hamiltonian. 

In the next section we introduce the abstract algebra behind our model and we give two concrete
representations. Then, in Subsection \ref{subsec:main-results}, we collect the key results and give a map to this
paper. Our conclusions are presented 
in Section \ref{sec:concl} together with some open problems.

\bigskip
{\bf The use of AI in this work.}

We used ChatGPT 5.4 and 5.5 Pro as a research assistant. Most computations were
performed by the AI, and were checked later by the human author. Besides analytic checks, we also 
performed a numerical test: an exact
diagonalization routine was programmed and run by the human author, completely independently of the numerical work of the
AI. We found agreement with the numerical results of the AI
and we independently confirmed 
the free-fermionic spectrum for generic coupling constants.

The text was written in a collaboration between the human and the AI. The Abstract, the Introduction, and the Conclusions
were written exclusively by the human author.

It may be of interest to spell out what the actual substantial contributions of the AI were and which
points required the insight of the human author. We present a list of these points in a
Supplemental Material\footnote{It can be downloaded as a separate PDF file from the preprint server.}.

\section{The model and the main results}
\label{sec:ising-algebra}

\subsection{The Ising algebra and the Hamiltonians}

\label{subsec:dressed-ising-algebra}

We define our models via an abstract algebra. 
This algebra is a double tensor product of two identical Ising-type algebras. We now give the precise definitions.

Fix an integer $M\geq 2$. We have two families of generators $B_j$ and $\tilde B_j$ with $M$ elements each, satisfying
\begin{equation}
B_j=B_j^\dagger,
\qquad
\tB_j=\tB_j^\dagger,
\qquad
j=1,\dots,M.
\end{equation}
Their squares are proportional to the identity:
\begin{equation}
B_j^2=\beta_j^2\I,
\qquad
\tB_j^2=\tbeta_j^2\I.
\label{eq:dressed-squares}
\end{equation}
Here $\beta_j$ and $\tbeta_j$ are real non-zero parameters.

The two families of generators satisfy the same algebra:
\begin{align}
\{B_j,B_{j+1}\}&=0,
&[B_j,B_k]&=0\qquad (|j-k|>1),
\label{eq:B-path-relations}
\\
\{\tB_j,\tB_{j+1}\}&=0,
&[\tB_j,\tB_k]&=0\qquad (|j-k|>1).
\label{eq:tB-path-relations}
\end{align}
It is important that there is no periodicity assumed. Thus each family is an open Ising-type algebra. The two families
commute with 
one another:
\begin{equation}
  [B_j,\tB_k]=0
\qquad \text{for all }j,k.
\label{eq:cross-B-commuting}
\end{equation}
We also introduce special 
cubic
products, which couple the two algebras:
\begin{equation}
A_j:=\ii\lambda_j B_jB_{j+1}\tB_j,
\qquad
\tA_j:=-\ii\tlambda_j\tB_j\tB_{j+1}B_{j+1},
\qquad (j=1,\dots,M-1),
\label{Adef}
\end{equation}
where $\lambda_j,\tlambda_j\in\mathbb R$ and we assume again that all of them are non-zero.
The factors $\ii$ and $-\ii$ make
these operators Hermitian, because $B_jB_{j+1}$ and
$\tB_j\tB_{j+1}$ are anti-Hermitian.  The squares of $A_j$ and $\tA_j$ are also scalar.  We write
\begin{equation}
A_j^2=\alpha_j^2\I,
\qquad
\tA_j^2=\talpha_j^2\I,
\label{eq:A-squares-alpha}
\end{equation}
where we introduced
\begin{equation}
\alpha_j=\lambda_j\beta_j\beta_{j+1}\tbeta_j,
\qquad
\talpha_j=\tlambda_j\tbeta_j\tbeta_{j+1}\beta_{j+1}.
\label{eq:pauli-alpha-defs}
\end{equation}
We define two Hamiltonians
\begin{equation}
H_1=\sum_{j=1}^{M}B_j+\sum_{j=1}^{M-1}A_j,
\qquad
H_2=\sum_{j=1}^{M}\tB_j+\sum_{j=1}^{M-1}\tA_j.
\label{eq:H1H2-algebraic}
\end{equation}
As final Hamiltonians we will consider the combinations $H_1\pm H_2$. 

The most important case of the final Hamiltonian is obtained when all cubic coefficients
collapse to one global parameter,
\begin{equation}
\lambda_j=\tlambda_j=\lambda
\qquad (j=1,\dots,M-1).
\label{eq:global-lambda}
\end{equation}
Then
\begin{equation}
A_j=\ii\lambda B_jB_{j+1}\tB_j,
\qquad
\tA_j=-\ii\lambda\tB_j\tB_{j+1}B_{j+1}.
\label{eq:cubic-A-global}
\end{equation}
In this case the model has $2M+1$ free parameters, given by the Ising-type couplings $\beta_j$, $\tilde\beta_j$ and the
global parameter $\lambda$. We could remove an overall normalization factor (for example by setting $\lambda=1$).
However, it is convenient to keep these parameters at the moment.

\subsection{Frustration graphs and simplicial cliques}

The notion of a frustration graph will be used repeatedly below. Frustration graphs can be introduced for sets of Pauli
strings (products of Pauli matrices acting on the Hilbert space of a lattice spin model). 
A Pauli string always squares to $\pm 1$, and two Pauli strings either commute or anti-commute. The commutation rules
can be encoded conveniently by a graph: we draw a vertex for each Pauli string, and for a pair of Pauli strings we draw
an edge if and only if they anti-commute.

These algebraic properties can be generalized to a more abstract setting, opening the way to ``graph-Clifford algebras'' \cite{graph-clifford,sajat-solving-ffd-1}.
Let us assume that we are dealing with an abstract algebra, where every generator squares to the identity, and two
generators either commute or anti-commute. Then we encode these properties into the frustration graph in the same way as
for the concrete Pauli strings: We draw vertices for the generators, and we connect two vertices if and only if the two generators
anti-commute. As an illustration, a frustration graph corresponding to the algebra of one family of the $B_j$
operators is shown in Figure \ref{fig:fr0}.

Frustration graphs play a central role in establishing the free-fermionic properties (and more generally: integrability)
in the lattice models that we study.

In the original definition of graph-Clifford algebras the generators are algebraically independent of each
other. However, the frustration graph can be introduced even in those cases when there are relations between the
generators. The effects of the extra relations will be discussed below.
If we
treat the terms in $H_1$ and $H_2$ as independent, then the frustration graphs for both Hamiltonians are as shown in
Fig. \ref{fig:fr1}.

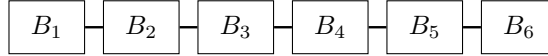
\begin{figure}[htbp]
  \centering
  \begin{tikzpicture}[x=1.25cm,y=1.25cm,font=\small,
ZY/.style={rectangle,draw,minimum width=10mm,minimum height=7mm,inner sep=1pt},
BB/.style={line width=0.9pt}
]

\foreach \i in {1,2,3,4,5,6} {
\node[ZY] (B\i) at ({\i+0.5},0) {$B_{\i}$};
}

\draw[BB] (B1) -- (B2);
\draw[BB] (B2) -- (B3);
\draw[BB] (B3) -- (B4);
\draw[BB] (B4) -- (B5);
\draw[BB] (B5) -- (B6);

\end{tikzpicture}

\caption{Frustration graph for one copy of the Ising-type algebra for $M=6$.}
  \label{fig:fr0}
\end{figure}

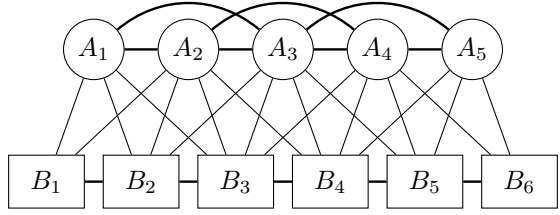
\begin{figure}[t]
  \centering
  \begin{tikzpicture}[x=1.25cm,y=1.25cm,font=\small,
XYY/.style={circle,draw,minimum size=8mm,inner sep=0pt},
ZY/.style={rectangle,draw,minimum width=10mm,minimum height=7mm,inner sep=1pt},
AA/.style={line width=0.9pt},
BB/.style={line width=0.9pt},
AB/.style={line width=0.4pt}
]

\foreach \i in {1,2,3,4,5} {
\node[XYY] (A\i) at ({\i+1},1.4) {$A_{\i}$};
}
\foreach \i in {1,2,3,4,5,6} {
\node[ZY] (B\i) at ({\i+0.5},0) {$B_{\i}$};
}

\draw[AA] (A1) -- (A2);
\draw[AA] (A2) -- (A3);
\draw[AA] (A3) -- (A4);
\draw[AA] (A4) -- (A5);
\draw[AA,bend left=40] (A1) to (A3);
\draw[AA,bend left=40] (A2) to (A4);
\draw[AA,bend left=40] (A3) to (A5);

\draw[BB] (B1) -- (B2);
\draw[BB] (B2) -- (B3);
\draw[BB] (B3) -- (B4);
\draw[BB] (B4) -- (B5);
\draw[BB] (B5) -- (B6);

\draw[AB] (A1) -- (B1);
\draw[AB] (A1) -- (B2);
\draw[AB] (A1) -- (B3);

\draw[AB] (A2) -- (B1);
\draw[AB] (A2) -- (B2);
\draw[AB] (A2) -- (B3);
\draw[AB] (A2) -- (B4);

\draw[AB] (A3) -- (B2);
\draw[AB] (A3) -- (B3);
\draw[AB] (A3) -- (B4);
\draw[AB] (A3) -- (B5);

\draw[AB] (A4) -- (B3);
\draw[AB] (A4) -- (B4);
\draw[AB] (A4) -- (B5);
\draw[AB] (A4) -- (B6);

\draw[AB] (A5) -- (B4);
\draw[AB] (A5) -- (B5);
\draw[AB] (A5) -- (B6);
\end{tikzpicture}

\caption{Frustration graph for the Hamiltonian $H_1$ for $M=6$. The terms $A_k$ are algebraically dependent on the
  generators $B_j$ and $\tilde B_j$; however, if we focus on $H_1$ alone, then we can regard them as independent
  operators, and thus the corresponding independent vertices are added to the graph.
  The Hamiltonian $H_2$ has an identical frustration graph. }
    \label{fig:fr1}
\end{figure}

\subsection{Representation by coupled Ising chains}
\label{subsec:ising-representation}

A faithful local representation
of the two commuting Ising algebras is obtained on two auxiliary Ising chains
with Pauli operators $\sigma_r^\alpha$ and $\tau_r^\alpha$:
\begin{equation}
B_{2r-1}\mapsto \beta_{2r-1}\sigma_r^z,
\qquad
B_{2r}\mapsto \beta_{2r}\sigma_r^x\sigma_{r+1}^x.
\label{eq:ising-B-map}
\end{equation}
\begin{equation}
\tB_{2r-1}\mapsto \tbeta_{2r-1}\tau_r^z,
\qquad
\tB_{2r}\mapsto \tbeta_{2r}\tau_r^x\tau_{r+1}^x.
\label{eq:ising-tB-map}
\end{equation}
The auxiliary length may be chosen as $\lceil(M+1)/2\rceil$.  Each Ising
algebra is represented by the usual alternating field/bond algebra of an
open Ising chain.

The cubic products then give the three-site generators
\begin{align}
A_{2r-1}&\mapsto -\alpha_{2r-1}\,
\sigma_r^y\sigma_{r+1}^x\tau_r^z,
&
  A_{2r}&\mapsto\alpha_{2r}
          \,
\sigma_r^x\sigma_{r+1}^y\tau_r^x\tau_{r+1}^x,
\label{eq:ising-A-map}
\\
  \tA_{2r-1}&\mapsto \tilde\alpha_{2r-1}
              \,
\sigma_r^x\sigma_{r+1}^x\tau_r^y\tau_{r+1}^x,
&
  \tA_{2r}&\mapsto  -\tilde\alpha_{2r}
            \,
\sigma_{r+1}^z\tau_r^x\tau_{r+1}^y.
\label{eq:ising-tA-map}
\end{align}
Thus the algebraic Hamiltonian is represented as a local two-leg Ising
ladder with open boundary conditions. The total number of qubits used in the construction is $2\lceil(M+1)/2\rceil$.

\subsection{Representation from the FFD perturbation of an XY chain}
\label{subsec:xy-representation}

We now use the alternative notation $X_j, Y_j, Z_j$ for the Pauli matrices acting on site $j$.
The algebra above has the following realization on a spin chain of
length $L=M+1$:
\begin{equation}
B_j=\beta_jY_jX_{j+1},
\qquad
\tB_j=\tbeta_jX_jY_{j+1},
\label{eq:pauli-B-realization}
\end{equation}
and
\begin{equation}
A_j=\alpha_j\,Z_jX_{j+1}X_{j+2},
\qquad
\tA_j=\tilde\alpha_j\,X_jX_{j+1}Z_{j+2}.
\label{eq:pauli-A-realization}
\end{equation}

It is worthwhile to consider the antisymmetric combination $H_1-H_2$. Focusing on the homogeneous version, where we set
$\beta_j=\tbeta_j=1$ and $\lambda_j=\tlambda_j=\lambda$, we get
\begin{equation}
H_1-H_2=
\sum_{j=1}^{L-1}(Y_jX_{j+1}-X_jY_{j+1})
+\lambda\sum_{j=1}^{L-2}
\left(Z_jX_{j+1}X_{j+2}-X_jX_{j+1}Z_{j+2}\right).
\label{eq:HQ-xy-perturbation}
\end{equation}
Here the nearest-neighbour part is a Jordan-Wigner solvable, $U(1)$-symmetric XY
chain. More precisely, it is given by the Dzyaloshinskii–Moriya interaction term, which can be transformed to the
standard form by a homogeneous twist transformation. The three-site interaction in the Hamiltonian above is the
antisymmetric combination of the terms 
 found in Fendley's free-fermion-in-disguise
construction \cite{ffd}. 

We note that the two representations have the same qubit count for odd $M$, whereas for even $M$ the Ising
representation has one more qubit than the XY representation.

\subsection{Main results}
\label{subsec:main-results}

This article proves the following statements for the Hamiltonians $H_{1,2}$ without using either concrete
representation in the proofs.

\begin{enumerate}[label=(\roman*),itemsep=0.4em]
\item The local cross-relations of the cubic algebra imply a complete
commutativity criterion.  Namely, we find the commutativity
$[H_1,H_2]=0$ if and only if
\begin{equation}
\lambda_j=\tlambda_j=\lambda_{j+1},
\end{equation}
for all indices for which the equation is defined.  Hence the generic
commuting family is the global-$\lambda$ family
\eqref{eq:cubic-A-global}. This is shown in Section \ref{sec:local-algebra}.

We stress that the terms in $H_1$ and $H_2$ do not all commute with one another, and the
special algebraic properties are needed for the cancellations. 

\item Each Hamiltonian is separately free-fermionic:
their frustration graphs are even-hole-free and
claw-free, so the free-fermion machinery of
\cite{fermions-behind-the-disguise,sajat-solving-ffd-1} applies.
More specifically, we can diagonalize the Hamiltonians as
\begin{equation}
  H_1=\sum_{k=1}^\cS \eps_{1,k} \mathcal{Z}_{1,k},\qquad H_2=\sum_{k=1}^\cS \eps_{2,k} \mathcal{Z}_{2,k},
\end{equation}
where $\cS=\lfloor (M+1)/2\rfloor$ is the number of fermionic eigenmodes, $\eps_{a,k}$ with $a=1,2$ and
$k=1,\dots,\cS$ are the 
single-particle energies, and 
$\mathcal{Z}_{a,k}$ are mode occupation operators with the properties
\begin{equation}
  \mathcal{Z}_{a,k}^2=1,\qquad [\mathcal{Z}_{a,k},\mathcal{Z}_{a,\ell}]=0.
\end{equation}
These operators are constructed from fermionic ladder operators $\Psi_{a,\pm k}$ as
\begin{equation}
  \mathcal{Z}_{a,k}=[\Psi_{a,k},\Psi_{a,-k}].
\end{equation}
The ladder operators satisfy the canonical anti-commutation algebra, which in our convention reads
\begin{equation}
  \{\Psi_{a,k},\Psi_{a,\ell}\}=\delta_{k+\ell}.
\end{equation}
We also construct transfer matrices $T_{1,M}(u)$ and $T_{2,M}(u)$ satisfying
\begin{equation}
  [T_{1,M}(u),T_{1,M}(v)]=0,\qquad [T_{2,M}(u),T_{2,M}(v)]=0.
\end{equation}
This is discussed in Section \ref{sec:independent-ff}.

\item 
The free-fermionic structures of the two Hamiltonians are also mutually commutative:
  \begin{equation}
    [\mathcal{Z}_{1,k},\mathcal{Z}_{2,\ell}]=0.
  \end{equation}
This is established using the commutativity of $H_1$ and $H_2$ and their free-fermionic decomposition. Furthermore, we
prove the cross-commutativity of the transfer matrices
\begin{equation}
[T_{1,M}(u),T_{2,M}(v)]=0.
\end{equation}
The proofs are presented in Subsection \ref{sec:cross1}.

\item We also treat the compatibility of the ladder operators.
After adjoining compatible edge operators, the two hidden Dirac
families can be chosen either mutually commuting or mutually anti-commuting.
This is presented in Subsection \ref{subsec:compatible-edges}.

\item In the spatially homogeneous case, the one-particle
  energies are encoded by an explicit scalar polynomial. We find the explicit recursion relation for this polynomial and
  study the large-volume limits. This leads to standing-wave equations for the single-particle modes. 
  In two special limits we obtain the one-particle energies of the Ising model and the FFD model,
  respectively. This is presented in Section \ref{sec:homogeneous-dispersion}.

\end{enumerate}

In Section \ref{sec:degs} we discuss the degeneracies of the spectra in these models. The actual degeneracies
depend on the 
representation chosen.
The final statement for generic couplings is as follows:
If $M$ is odd, then all
energy levels are non-degenerate in both representations. If $M$ is even, then there is a homogeneous degeneracy of 4 in
the Ising representation and 2 in the XY representation. Numerical data for the XY case is presented in Appendix
\ref{app:numerics}. 

In two further appendices we present two computations that are not crucial for our main line of argument, but further
support the claim that our model is indeed a natural extension of Fendley's FFD model:
\begin{enumerate}[label=(\roman*),itemsep=0.4em]
\item The highly non-local transfer matrices can be factorized into a product of local operators, 
  similar to the factorization derived earlier in \cite{ffd} (see also \cite{sajat-floquet}). This is presented
  in Appendix \ref{app:factor}.
\item 
  The transfer matrices and the fermionic operators can be expressed as Matrix Product Operators (MPOs) with
  fixed small bond dimension. We treat only the XY representation. 
In that representation the transfer matrix admits an MPO representation with bond dimension 3,
so that this MPO is a simple modification of the one found by Fendley \cite{ffd}. The fermionic
operators can be represented as an inhomogeneous MPO of bond dimension 4. This appears to be a new result even for the
original FFD model. These computations are presented in
Appendix \ref{app:MPO}.
\end{enumerate}

\section{Local algebra and mutual commutativity}
\label{sec:local-algebra}

Here we investigate the commutation relations between the terms in the Hamiltonians $H_1$ and $H_2$. We determine the
frustration graphs that we would obtain if the $A_k$ and $\tilde A_k$ operators were algebraically independent
generators. We also establish the sufficient and necessary condition for the global commutativity of $H_1$ and $H_2$.

\subsection{Same-family relations and frustration graphs}
\label{subsec:same-family-relations}

The definitions \eqref{Adef} determine all the commutation relations inside a
single Hamiltonian.

\begin{lemma}[Internal sign rules]
\label{lem:internal-sign-rules}
In the $H_1$ family, the only anti-commuting pairs are
\begin{align}
\{B_j,B_k\}=0 &\quad \Longleftrightarrow\quad |j-k|=1,
\label{eq:BB-internal}
\\
\{A_j,B_k\}=0 &\quad \Longleftrightarrow\quad k\in\{j-1,j,j+1,j+2\},
\label{eq:AB-internal}
\\
\{A_j,A_k\}=0 &\quad \Longleftrightarrow\quad 1\leq |j-k|\leq 2.
\label{eq:AA-internal}
\end{align}
All other same-family pairs commute.  The same statements hold for the
tilded family.
\end{lemma}

\begin{proof}
Only the Ising-algebra signs are used.  The relation for $B_j,B_k$ is one of
the defining relations.  Since
$A_j=\ii\lambda_jB_jB_{j+1}\tB_j$ and all $B$'s commute with all $\tB$'s,
commuting $B_k$ through $A_j$ only counts the signs produced by
$B_jB_{j+1}$.  There is one such sign exactly for
$k\in\{j-1,j,j+1,j+2\}$, which proves \eqref{eq:AB-internal}.

For $A_j$ and $A_k$, the $B$-part contributes the parity of the adjacent
crossings between the two sets $\{j,j+1\}$ and $\{k,k+1\}$, while the
$\tB$-part contributes the parity of the crossing between $\tB_j$ and
$\tB_k$.  The total parity is odd exactly when the intervals
$[j,j+2]$ and $[k,k+2]$ overlap without being identical, i.e. when
$1\leq |j-k|\leq 2$.  This proves \eqref{eq:AA-internal}.  The tilded
case is identical with the two Ising chains exchanged.
\end{proof}

From these relations we can draw the frustration graphs. We obtain examples for the so-called {interval graphs}. An
interval graph is defined as follows: vertices correspond to intervals on the real line, and we draw edges between two
vertices if the corresponding two intervals overlap. Let us now introduce the following intervals for the generators in
$H_1$:
\begin{equation}
I(B_j)=[j,j+1],
\qquad
I(A_j)=[j,j+2].
\label{eq:formal-intervals}
\end{equation}
We can see that the frustration graph of $H_1$ is the interval graph of the intervals
\eqref{eq:formal-intervals}.  The same graph governs $H_2$. As an illustration, see Figure \ref{fig:fr1}, which depicts
the graph for $M=6$.

In a later section, an important role will be played by the so-called independence number of the graphs, which is defined
as the maximal size of a subset of the vertices such that no two vertices are connected. It is easy to see that in these
graphs the independence number is $\cS=\left\lfloor\frac{M+1}{2}\right\rfloor$.

\subsection{Cross-relations and commutativity criterion}
\label{subsec:cross-relations}

We now list the commutation relations between the terms in the two different Hamiltonians.

\begin{lemma}[Algebraic cross-relations]
\label{lem:algebraic-cross-relations}
For all allowed indices,
\begin{equation}
[B_j,\tB_k]=0,
\qquad
[A_j,\tA_k]=0.
\label{eq:AAtilde-BBtilde-commute}
\end{equation}
The only mixed anti-commuting pairs are
\begin{equation}
\{A_j,\tB_{j-1}\}=0,
\qquad
\{A_j,\tB_{j+1}\}=0,
\qquad
\{B_j,\tA_j\}=0,
\qquad
\{B_{j+2},\tA_j\}=0.
\label{eq:mixed-anti-commuting-pairs}
\end{equation}
Moreover, the paired products have the following common normal forms:
\begin{align}
A_j\tB_{j+1}
&=\ii\lambda_j\,B_jB_{j+1}\tB_j\tB_{j+1},
&
B_j\tA_j
&=-\ii\tlambda_j\,B_jB_{j+1}\tB_j\tB_{j+1},
\label{eq:weighted-cross-square-1}
\\
A_{j+1}\tB_j
&=-\ii\lambda_{j+1}\,B_{j+1}B_{j+2}\tB_j\tB_{j+1},
&
B_{j+2}\tA_j
&=\ii\tlambda_j\,B_{j+1}B_{j+2}\tB_j\tB_{j+1}.
\label{eq:weighted-cross-square-2}
\end{align}
In the global-$\lambda$ family these imply the exact square identities
\begin{equation}
A_j\tB_{j+1}=-B_j\tA_j,
\qquad
A_{j+1}\tB_j=-B_{j+2}\tA_j.
\label{eq:global-cross-square-identities}
\end{equation}
\end{lemma}

\begin{proof}
The relation $[B_j,\tB_k]=0$ is part of the defining algebra.  For
$[A_j,\tA_k]$, move the $B$-part $B_jB_{j+1}$ through the single $B$-factor
$B_{k+1}$ in $\tA_k$, and move the single $\tB$-factor $\tB_j$ through the
$\tB$-part $\tB_k\tB_{k+1}$.  The two parities agree for every $j,k$:
for $k=j-2,j-1,j,j+1$ both parities are odd, and otherwise both are even.
Hence the total parity is always even.

The list \eqref{eq:mixed-anti-commuting-pairs} follows by the same sign
count.  For example, $A_j$ contains one tilded factor, namely $\tB_j$, so it
anti-commutes with $\tB_k$ precisely when $k=j\pm1$.  Similarly, $\tA_j$
contains the single untilded factor $B_{j+1}$, so it anti-commutes with
$B_k$ precisely when $k=j$ or $k=j+2$.

Finally, putting the paired products into the same normal order gives
\begin{align}
A_j\tB_{j+1}
&=\ii\lambda_jB_jB_{j+1}\tB_j\tB_{j+1},
\\
B_j\tA_j
&=-\ii\tlambda_jB_j\tB_j\tB_{j+1}B_{j+1}
=-\ii\tlambda_jB_jB_{j+1}\tB_j\tB_{j+1},
\end{align}
which proves \eqref{eq:weighted-cross-square-1}.  Likewise,
\begin{align}
A_{j+1}\tB_j
&=\ii\lambda_{j+1}B_{j+1}B_{j+2}\tB_{j+1}\tB_j
=-\ii\lambda_{j+1}B_{j+1}B_{j+2}\tB_j\tB_{j+1},
\\
B_{j+2}\tA_j
&=-\ii\tlambda_jB_{j+2}\tB_j\tB_{j+1}B_{j+1}
=\ii\tlambda_jB_{j+1}B_{j+2}\tB_j\tB_{j+1},
\end{align}
which proves \eqref{eq:weighted-cross-square-2}.
\end{proof}

As an application we have the following:

\begin{theorem}[Algebraic mutual commutativity]
\label{thm:mutual-commutativity}
We have
\begin{equation}
[H_1,H_2]=0
\end{equation}
if and only if
\begin{equation}
\lambda_j=\tlambda_j
\qquad (1\leq j\leq M-1)
\label{eq:lambda-tilde-same-index}
\end{equation}
and
\begin{equation}
\lambda_{j+1}=\tlambda_j
\qquad (1\leq j\leq M-2).
\label{eq:lambda-tilde-shifted-index}
\end{equation}
Equivalently, all local coefficients are equal to one global parameter
$\lambda$.
\end{theorem}

\begin{proof}
By Lemma~\ref{lem:algebraic-cross-relations}, the only non-zero mixed
commutators are the four families in
\eqref{eq:mixed-anti-commuting-pairs}.  Since each of these pairs
anti-commutes, its commutator is twice its ordered product.  Hence the
commutator is the sum of the local obstructions
\begin{equation}
2\left(A_j\tB_{j+1}+B_j\tA_j\right)
\end{equation}
and
\begin{equation}
2\left(A_{j+1}\tB_j+B_{j+2}\tA_j\right),
\end{equation}
with boundary terms included only when the indices exist.  Expanding the cubic
definitions
gives
\begin{align}
2\left(A_j\tB_{j+1}+B_j\tA_j\right)
&=
2\ii\left(\lambda_j-\tlambda_j\right)
B_jB_{j+1}\tB_j\tB_{j+1},
\label{eq:first-obstruction}
\\
2\left(A_{j+1}\tB_j+B_{j+2}\tA_j\right)
&=
2\ii\left(\tlambda_j-\lambda_{j+1}\right)
B_{j+1}B_{j+2}\tB_j\tB_{j+1}.
\label{eq:second-obstruction}
\end{align}
Therefore \eqref{eq:lambda-tilde-same-index} and
\eqref{eq:lambda-tilde-shifted-index} are sufficient for commutativity.  In the
abstract algebra the normal monomials in \eqref{eq:first-obstruction} and
\eqref{eq:second-obstruction} are independent, so their coefficients must
vanish.  This gives
$\lambda_j=\tlambda_j$ and $\lambda_{j+1}=\tlambda_j$ and proves necessity.
\end{proof}

\section{Independent free-fermion structure}
\label{sec:independent-ff}

We now investigate the two Hamiltonians separately and show that they are integrable and free-fermionic.

We will use the techniques of \cite{ffd,fermions-behind-the-disguise}.
In this section we will proceed as if the terms of $H_1$ and $H_2$ were algebraically independent, and then we can use
the calculations of \cite{ffd,fermions-behind-the-disguise} directly. The key point is that we focus on either $H_1$ or
$H_2$, but here we will not yet treat their linear combinations. This implies that we can choose either the set
$\{\{B_j\}_{j=1,\dots,M},\{A_k\}_{k=1,\dots,M-1}\}$
or
$\{\{\tilde B_j\}_{j=1,\dots,M},\{\tilde A_k\}_{k=1,\dots,M-1}\}$
as independent generators.

This procedure is justified as long as we do not consider $H_1$ and $H_2$ at the same time (we do not diagonalize their
linear combinations).

\subsection{Free fermions behind the disguise}
\label{subsec:ffbd-general-construction}

We now recall the key results of \cite{ffd,fermions-behind-the-disguise}.  The reader who is familiar with this
construction may skip this subsection.

Let
\begin{equation}
H=\sum_{v\in V} g_v
\label{eq:general-graph-H}
\end{equation}
be a Hamiltonian whose terms are Hermitian graph-Clifford generators \cite{sajat-solving-ffd-1},
\begin{equation}
g_v=g_v^\dagger,\qquad g_v^2=w_v^2\I ,
\end{equation}
and assume that each pair of distinct generators either commutes or
anti-commutes.  The frustration graph \(G=G(H)\) has vertex set \(V\), with
an edge between \(v\) and \(v'\) precisely when
\begin{equation}
\{g_v,g_{v'}\}=0 .
\end{equation}
Thus non-adjacent vertices correspond to commuting generators.

A hole is an induced cycle of length at least four, and an even hole is a
hole with an even number of vertices.  A claw is an induced copy of
the graph \(K_{1,3}\) (see Fig. \ref{fig:claw}).  We say that \(G\) is ECF if it is both even-hole-free and
claw-free.  We shall also use a simplicial clique \(K\subset V\): this is a
non-empty clique such that, for every \(v\in K\), the set
\(N_G(v)\setminus K\) induces a clique.  For ECF graphs such a simplicial
clique exists, and it gives an edge operator \(\chi\) satisfying
\begin{equation}
\chi=\chi^\dagger,\qquad \chi^2=\I,\qquad
\{\chi,g_v\}=0\quad (v\in K),\qquad
[\chi,g_v]=0\quad (v\notin K).
\label{eq:general-edge-operator}
\end{equation}
The operator \(\chi\) may be adjoined algebraically if it is not already
present in the representation.

\begin{figure}[t]
  \centering
  \begin{tikzpicture}[
        scale=0.5,
    transform shape,
    vertex/.style={circle, draw, fill=white, inner sep=2pt, minimum size=7mm},
    edge/.style={thick}
]

\node[vertex] (c) at (0,0) {};
\node[vertex] (a) at (-1.5,1.2) {};
\node[vertex] (b) at (1.5,1.2) {};
\node[vertex] (d) at (0,-1.5) {};

\draw[edge] (c) -- (a);
\draw[edge] (c) -- (b);
\draw[edge] (c) -- (d);

\end{tikzpicture}
  
  \caption{The claw graph}
  \label{fig:claw}
\end{figure}
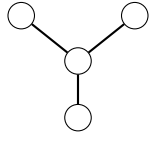

Let \(\cI(G)\) be the set of independent sets of \(G\), and let
\(\cI_r(G)\subset\cI(G)\) denote the independent sets of size \(r\).  We write
\[
\cS(G)=\max_{S\in\cI(G)}|S|
\]
for the independence number.  For \(S\in\cI(G)\) define
\begin{equation}
g_S=\prod_{v\in S}g_v,\qquad g_\emptyset=\I .
\end{equation}
The order in this product is immaterial, because all generators belonging to
an independent set commute.  The independent-set charges are
\begin{equation}
Q_G^{(r)}=\sum_{S\in\cI_r(G)}g_S,
\qquad
r=0,\ldots,\cS(G),
\label{eq:general-independent-set-charges}
\end{equation}
with \(Q_G^{(0)}=\I\) and \(Q_G^{(1)}=H\).  If \(G\) is claw-free, then
\begin{equation}
[Q_G^{(r)},Q_G^{(s)}]=0
\qquad
(r,s=0,\ldots,\cS(G)).
\label{eq:general-independent-set-commuting}
\end{equation}
Equivalently, the transfer matrix
\begin{equation}
T_G(u)=\sum_{r=0}^{\cS(G)}(-u)^r Q_G^{(r)}
      =\sum_{S\in\cI(G)}(-u)^{|S|}g_S
\label{eq:general-graph-transfer-matrix}
\end{equation}
forms a commuting family,
\begin{equation}
[T_G(u),T_G(v)]=0 .
\end{equation}
It is normalized by
\begin{equation}
T_G(0)=\I,\qquad -\partial_uT_G(0)=H .
\end{equation}
The weighted independence
polynomial is
\begin{equation}
P_G(x)=
\sum_{S\in\cI(G)}
(-x)^{|S|}
\prod_{v\in S}w_v^2 .
\label{eq:general-weighted-independence-polynomial}
\end{equation}
If \(G\) is ECF, the transfer matrix obeys the inversion relation
\begin{equation}
T_G(u)T_G(-u)=P_G(u^2)\I .
\label{eq:general-graph-inversion}
\end{equation}
For generic non-zero weights we write
\begin{equation}
P_G(x)=\prod_{k=1}^{\cS(G)}\left(1-\eps_k^2x\right),
\qquad
u_k=\eps_k^{-1}.
\label{eq:general-polynomial-factorization}
\end{equation}
The numbers \(\eps_k\) are the one-particle energies.

The hidden fermionic operators are obtained by dressing the edge operator
with the transfer matrix:
\begin{equation}
\Psi_k
=
C_k\,T_G(-u_k)\chi T_G(u_k),
\qquad
\Psi_{-k}
=
C_k\,T_G(u_k)\chi T_G(-u_k),
\qquad
k=1,\ldots,\cS(G).
\label{eq:general-hidden-fermions}
\end{equation}
The constants \(C_k\) are chosen so that the canonical anti-commutation
relations hold,
\begin{equation}
\{\Psi_k,\Psi_\ell\}=\delta_{k+\ell}\I .
\label{eq:general-CAR}
\end{equation}
The normalization factor is
\begin{equation}
  \label{norma}
C_k^{-2}
=
-16u_k^2\,
P_{G\setminus K}(u_k^2)\,
P_G'(u_k^2),
\end{equation}
where \(P_{G\setminus K}\) is the polynomial
\eqref{eq:general-weighted-independence-polynomial} for the induced graph
obtained by deleting the simplicial clique \(K\).

The operators \(\Psi_{\pm k}\) are ladder operators for \(H\):
\begin{equation}
[H,\Psi_{\pm k}]=\pm 2\eps_k\,\Psi_{\pm k} .
\label{eq:general-ladder-relation}
\end{equation}
Moreover, the Hamiltonian itself is reconstructed as the free-fermion
quadratic form
\begin{equation}
H=\sum_{k=1}^{\cS(G)}
\eps_k[\Psi_k,\Psi_{-k}].
\label{eq:general-H-reconstruction}
\end{equation}
Thus
\begin{equation}
\Spec(H)
=
\left\{
\sum_{k=1}^{\cS(G)}s_k\eps_k
\;:\;
s_k=\pm1
\right\},
\label{eq:general-free-spectrum}
\end{equation}
up to the homogeneous degeneracy coming from degrees of freedom not resolved
by these \(\cS(G)\) hidden fermions.

\subsection{The origin of exponential degeneracies}

\label{sec:origin}

Here we explain why it is difficult to find FFD-type models that are free from exponential degeneracies, based on the
results quoted in the last subsection. The key observation is that in such models the number of fermionic eigenmodes is
given by $\cS$, the independence number of the frustration graph, which is typically smaller than the number of qubits
needed for the representations of the algebras.

If a graph-Clifford algebra has $M$ generators, and it
describes a lattice spin model, then $\cS\sim M$. In a connected claw-free graph the
maximal possible asymptotic $\cS/M$ ratio is $1/2$. This is a classic result in graph theory
\cite{Sumner1974GraphsWith1Factors,LasVergnas1975NoteMatchings}: the maximum value of $\cS$ for a connected claw-free
graph is $\lceil M/2 \rceil$.
It follows that the maximal number of
fermionic eigenmodes is $\lceil M/2 \rceil$.
In the models that appeared in the literature the bound was reached only by the Ising/XY families, whose frustration
graphs are the path or circle graphs, for open or periodic boundary conditions, respectively.
In other FFD-type models that have appeared in the literature, the asymptotic $\cS/M$ ratio was smaller; in the case of
the FFD model, this ratio is 1/3. 

If a graph-Clifford algebra is non-degenerate (has a trivial center), or the dimension of the center is bounded, then
the algebra can be represented on $N=M/2+\ordo(1)$ qubits \cite{sajat-solving-ffd-1}. There are no smaller
representations.

This implies that in every case in which the asymptotic $\cS/M$ ratio is smaller than $1/2$, there will be exponential
degeneracies due
to the missing number of eigenmodes. In the case of the FFD model the degeneracies were recently tied to new families of
hidden fermions \cite{eric-lorenzo-ffd-1}. 

These arguments show that for many graph-Clifford algebras which are more complicated than the Ising-type algebras one
would need an extensive number of independent central elements, so that the dimension of the representations can be
lowered. In parallel, one might need to work with generalizations of the framework of
\cite{ffd,fermions-behind-the-disguise}, in order to avoid the strong requirements for the ECF graphs.

We will see that our model uses both mechanisms together to avoid the appearance of the exponential degeneracies. We
will return to this issue in the Conclusions, see Section \ref{sec:concl}.

Finally, we also remark that the upper bound for the independence number is reached in many other connected claw-free graphs
or graph families. However, the corresponding models might not have an interpretation as a spin-chain
model, or they might be related to Jordan-Wigner solvable models by simple unitary transformations.
This open direction is left for further research.

\subsection{Frustration graphs and simplicial cliques in our model}
\label{subsec:frustration-graphs}

We now apply the formalism presented above to $H_1$ and $H_2$.

Let $\Gamma_1$ be the graph whose vertices are the generators of $H_1$, with
an edge between two vertices exactly when the corresponding algebra elements
anti-commute.  Lemma~\ref{lem:internal-sign-rules} identifies $\Gamma_1$ with a particular
interval graph. The frustration graph for $H_2$ is identical.

Interval graphs are chordal, hence contain no induced cycle of length at
least four.  In the present interval family they are also claw-free: one
interval of length two or three cannot meet three pairwise disjoint intervals
of length at least two.  Therefore $\Gamma_1$ is ECF.

We now discuss the issue of the algebraic relations between the various terms of the two Hamiltonians. Our main
claim is that we can immediately apply the results of \cite{ffd,fermions-behind-the-disguise} if we treat $H_1$ and
$H_2$ independently. The reason for this is that although the operators \(A_j\) are products in the larger algebra, the
family \(\{B_j,A_j\}\) has no further algebraic relations.  This is shown as follows. After rescaling
the generators to involutions, a monomial
\(\prod_j B_j^{x_j}\prod_j A_j^{y_j}\), with \(x_j,y_j\in\{0,1\}\), has
tilded exponents \(y_j\) in the normal form of the two Ising algebras.
Hence a scalar monomial must have all \(y_j=0\), and then the untilded
exponents force all \(x_j=0\).  Thus the subalgebra generated by
\(\{B_j,A_j\}\) is the graph-Clifford algebra with frustration graph
\(\Gamma_1\).  The tilded family is analogous.

These arguments imply that the Hamiltonians $H_1$ and $H_2$ separately possess an extensive number of conserved
quantities, and they can be solved by free-fermionic operators.

The construction of \cite{ffd,fermions-behind-the-disguise} uses a so-called edge operator, which is associated with a
simplicial clique in the graph. The left boundary pair
\begin{equation}
K_1=\{A_1,B_1\}
\end{equation}
is a simplicial clique for $\Gamma_1$.  This means that we need to find an 
operator $\chi_1=\chi_1^\dagger$ with $\chi_1^2=\I$ such that
\begin{equation}
\{\chi_1,A_1\}=\{\chi_1,B_1\}=0,
\qquad
[\chi_1,g]=0
\quad \text{for all other generators }g\text{ of }H_1.
\label{eq:chi1-edge-relations}
\end{equation}
In the abstract treatment such an operator can be chosen as an additional generator, and in some cases it can be chosen
to lie within the algebra itself \cite{sajat-solving-ffd-1}. We do not discuss these subtleties at this point; instead
we state only that a suitable $\chi_1$ can be found in our concrete representations. This is enough to construct the free-fermionic
operators that solve the Hamiltonians.

For $H_2$ it is natural to choose the simplicial clique to lie at the other boundary:
\begin{equation}
K_2=\{\tA_{M-1},\tB_M\}.
\end{equation}
In this way, we can construct mutually commuting or anti-commuting ladder operators for the two Hamiltonians, see Section
\ref{sec:mutual-ff}. If we chose $K_1$ as the simplicial clique for $H_2$, then we would not be able to construct
cross-compatible operators for $\lambda\ne 0$, see the remark at the end of Section \ref{sec:mutual-ff}.

\subsection{Transfer matrices}
\label{subsec:transfer-matrices}

We now derive transfer matrices using the methods of \cite{ffd,fermions-behind-the-disguise}.
We will denote by $0<m\le M$ the subalgebras corresponding to smaller spin chains, and we will derive recursions in
$m$. The full chain corresponds to
$m=M$.

Consider $m$ short generators, and let $\cI_m$ be the set of independent
sets in the interval family
\[
\{I(B_j)\}_{j=1}^{m}\cup\{I(A_j)\}_{j=1}^{m-1}.
\]
For $m=0,1$ the second set is empty.
For $S\in\cI_m$, let $g_1(S)$ be the product of the corresponding
$H_1$ generators, ordered by increasing left endpoint.  Let $g_2(S)$ be the
same ordered product with all generators tilded.  The order is immaterial
inside an independent set, because its generators commute.

\begin{definition}
For $0\leq m\leq M$, define the prefix transfer matrices by
\begin{equation}
T_{1,m}(u)=\sum_{S\in\cI_m}(-u)^{|S|}g_1(S),
\qquad
T_{2,m}(u)=\sum_{S\in\cI_m}(-u)^{|S|}g_2(S).
\label{eq:FFBD-transfer-matrices}
\end{equation}
The full transfer matrices are the $m=M$ members.  They are normalized by
\begin{equation}
-\partial_uT_{1,M}(0)=H_1,
\qquad
-\partial_uT_{2,M}(0)=H_2.
\end{equation}
\end{definition}

\begin{remark}[]
  The order of the transfer matrices in $u$ is determined by the independence number of the frustration graph,
  $\cS=\left\lfloor\frac{M+1}{2}\right\rfloor$.
\end{remark}

\begin{proposition}[Boundary recursion]
\label{prop:transfer-recursion-new}
For $m\geq3$,
\begin{align}
T_{1,m}(u)&=T_{1,m-1}(u)-uB_mT_{1,m-2}(u)-uA_{m-1}T_{1,m-3}(u),
\label{eq:T1-recursion-new}
\\
T_{2,m}(u)&=T_{2,m-1}(u)-u\tB_mT_{2,m-2}(u)-u\tA_{m-1}T_{2,m-3}(u),
\label{eq:T2-recursion-new}
\end{align}
with initial conditions
\begin{equation}
T_{a,0}(u)=1,
\qquad
T_{1,1}(u)=1-uB_1,
\qquad
T_{2,1}(u)=1-u\tB_1,
\end{equation}
and
\begin{equation}
T_{1,2}(u)=1-u(B_1+B_2+A_1),
\qquad
T_{2,2}(u)=1-u(\tB_1+\tB_2+\tA_1).
\end{equation}
\end{proposition}

\begin{proof}
An independent set in $\cI_m$ is exactly one of the following: it uses no
generator touching the right endpoint of the chain, giving an independent set
in $\cI_{m-1}$; it uses the short generator $B_m$ (or $\tB_m$), in which case
the rest lies in $\cI_{m-2}$; or it uses the long generator $A_{m-1}$ (or
$\tA_{m-1}$), in which case the rest lies in $\cI_{m-3}$.  This is precisely
\eqref{eq:T1-recursion-new} and \eqref{eq:T2-recursion-new}.  The displayed
initial conditions are the direct independent-set sums for zero, one, and two
short generators.
\end{proof}

\subsection{Inversion relation and one-particle energies}
\label{subsec:inversion-independent}

For each family, the inversion relation is a purely algebraic
cancellation among independent-set monomials.

\begin{theorem}[Independent inversion relation]
\label{thm:independent-inversion}
For $a=1,2$,
\begin{equation}
T_{a,M}(u)T_{a,M}(-u)=P_{a,M}(u^2)\I,
\label{eq:independent-inversion}
\end{equation}
where $P_{a,M}$ is the weighted independence polynomial of order $\cS$, given by
\begin{equation}
P_{a,M}(x)=\sum_{S\in\cI_M}(-x)^{|S|}\prod_{g\in S}w_a(g)^2.
\label{eq:weighted-independence-general}
\end{equation}
Here $w_a(g)^2$ denotes the scalar square of the generator $g$ in the
$a$-th family.  Thus
\begin{align}
&w_1(B_j)^2=\beta_j^2,
&&w_1(A_j)^2=\alpha_j^2,
\label{eq:w1-weights}
\\
&w_2(\tB_j)^2=\tbeta_j^2,
&&w_2(\tA_j)^2=\talpha_j^2.
\label{eq:w2-weights}
\end{align}
\end{theorem}

\begin{proof}
As long as we treat the two transfer matrices separately, the formalism of \cite{ffd,fermions-behind-the-disguise}
applies directly. The frustration graphs of $H_1$ and $H_2$ are identical and this graph is ECF. The statement of the
theorem follows immediately.
\end{proof}

\begin{proposition}[Inhomogeneous recursion for $P_{a,m}$]
For the polynomials defined by \eqref{eq:weighted-independence-general}
with $M$ replaced by $m$, set
\[
P_{a,0}(x)=1 .
\]
Then
\[
P_{1,1}(x)=1-x\beta_1^2,
\qquad
P_{2,1}(x)=1-x\tbeta_1^2,
\]
and
\[
P_{1,2}(x)=1-x(\beta_1^2+\beta_2^2+\alpha_1^2),
\qquad
P_{2,2}(x)=1-x(\tbeta_1^2+\tbeta_2^2+\talpha_1^2).
\]
For $m\geq 3$ one has
\begin{align}
P_{1,m}(x)
&=
P_{1,m-1}(x)
-x\,\beta_m^{\,2}\,P_{1,m-2}(x)
-x\,\alpha_{m-1}^{\,2}\,
P_{1,m-3}(x),
\\
P_{2,m}(x)
&=
P_{2,m-1}(x)
-x\,\tbeta_m^{\,2}\,P_{2,m-2}(x)
-x\,\talpha_{m-1}^{\,2}\,
P_{2,m-3}(x).
\end{align}
\end{proposition}

\begin{proof}
Decompose an independent set according to the right boundary.  It either
contains no generator touching the right endpoint, giving $P_{a,m-1}$; or it
contains the short interval $B_m$ or $\widetilde B_m$, leaving an
independent set in $\cI_{m-2}$; or it contains the long interval $A_{m-1}$ or
$\widetilde A_{m-1}$, leaving an independent set in
$\cI_{m-3}$.  Multiplying by the corresponding scalar squares gives the
stated recursion.
\end{proof}

Consider the factorization
\begin{equation}
P_{a,M}(x)=\prod_{k=1}^{\cS}(1-\eps_{a,k}^2x).
\label{eq:P-roots-eps}
\end{equation}
We can now state the main result of this section: using the constructions of 
\cite{ffd,fermions-behind-the-disguise} we obtain free-fermionic operators $\Psi_{a,k}$ such that
\begin{equation}
H_a=\sum_k\eps_{a,k} \mathcal{Z}_{a,k},\qquad \mathcal{Z}_{a,k}\equiv [\Psi_{a,k},\Psi_{a,-k}].
\label{eq:H-FFBD-spectral}
\end{equation}
The corresponding
 Dirac operators are obtained from an edge operator by the standard
transfer-matrix formula
\begin{equation}
\Psi_{a,\pm k}=C_{a,k}\,T_{a,M}(\mp u_{a,k})\,\chi_a\,T_{a,M}(\pm u_{a,k}),
\qquad
u_{a,k}=\eps_{a,k}^{-1},
\label{eq:Dirac-transfer-formula-new}
\end{equation}
where $C_{a,k}$ is the normalization factor obtained from the general formula \eqref{norma}.

\section{Mutual free-fermion structure}
\label{sec:mutual-ff}

We now establish the compatibility of the fermionic structures constructed
earlier separately for $H_1$ and $H_2$.
From now on, we work in the global-$\lambda$ family
\eqref{eq:cubic-A-global}, so that $[H_1,H_2]=0$.

First we prove the commutativity of the mode occupation operators $\mathcal{Z}_{a,k}$ introduced in 
\eqref{eq:H-FFBD-spectral}.

Then we prove the cross-commutativity of the transfer matrices and show that we can choose
appropriate edge operators such that the fermionic ladder operators themselves are compatible.

\subsection{Cross-commutativity of the mode occupation operators}

\label{sec:cross1}

\begin{proposition}[Mode occupations as polynomials in one Hamiltonian]
\label{prop:mode-occupation-polynomial-in-H}
Fix \(a\in\{1,2\}\).  Assume that the same-family hidden fermions have already
been constructed, so that
\begin{equation}
H_a=\sum_{r=1}^{\cS}\eps_{a,r}\cZ_{a,r},
\qquad
\cZ_{a,r}=[\Psi_{a,r},\Psi_{a,-r}],
\label{eq:H-as-Z-sum}
\end{equation}
with
\begin{equation}
\cZ_{a,r}^2=\I,
\qquad
[\cZ_{a,r},\cZ_{a,s}]=0 .
\end{equation}
Assume also the generic non-resonance condition
\begin{equation}
E_a(\sigma)\neq E_a(\tau)
\qquad
\text{for all }\sigma\neq\tau,\qquad
\sigma,\tau\in\{\pm1\}^{\cS},
\label{eq:generic-sign-sum-nondegeneracy}
\end{equation}
where
\begin{equation}
E_a(\sigma)=\sum_{r=1}^{\cS}\sigma_r\eps_{a,r}.
\end{equation}
Then each mode occupation involution \(\cZ_{a,k}\) is a polynomial in \(H_a\).
More precisely,
\begin{equation}
\cZ_{a,k}=z_{a,k}(H_a),
\label{eq:Z-as-polynomial-in-H}
\end{equation}
where the interpolation polynomial is
\begin{equation}
z_{a,k}(x)
=
\sum_{\sigma\in\{\pm1\}^{\cS}}
\sigma_k
\prod_{\substack{\tau\in\{\pm1\}^{\cS}\\ \tau\neq\sigma}}
\frac{x-E_a(\tau)}{E_a(\sigma)-E_a(\tau)} .
\label{eq:Z-interpolation-polynomial}
\end{equation}
\end{proposition}

\begin{proof}
We first construct the projectors that select the different sign patterns. Since the operators \(\cZ_{a,r}\) are
commuting involutions, the operators 
\begin{equation}
\Pi_a(\sigma)
=
2^{-\cS}\prod_{r=1}^{\cS}
\left(\I+\sigma_r\cZ_{a,r}\right),
\qquad
\sigma\in\{\pm1\}^{\cS},
\label{eq:Z-joint-projectors}
\end{equation}
are mutually orthogonal idempotents and satisfy
\begin{equation}
\sum_{\sigma\in\{\pm1\}^{\cS}}\Pi_a(\sigma)=\I .
\end{equation}
Moreover,
\begin{equation}
\cZ_{a,k}\Pi_a(\sigma)=\sigma_k\Pi_a(\sigma).
\label{eq:H-Z-on-projectors}
\end{equation}
Therefore, using \eqref{eq:H-as-Z-sum}
we obtain
\[
H_a\Pi_a(\sigma)
=
\left(\sum_{r=1}^{\cS}\sigma_r\eps_{a,r}\right)\Pi_a(\sigma)
=
E_a(\sigma)\Pi_a(\sigma).
\]
Consequently
\begin{equation}
H_a
=
\sum_{\sigma\in\{\pm1\}^{\cS}}
E_a(\sigma)\Pi_a(\sigma).
\label{eq:H-spectral-resolution-Pi}
\end{equation}
More generally, for every polynomial \(f\),
\begin{equation}
f(H_a)
=
\sum_{\sigma\in\{\pm1\}^{\cS}}
f(E_a(\sigma))\Pi_a(\sigma).
\label{eq:polynomial-functional-calculus-Pi}
\end{equation}
Indeed, this follows from
\(H_a^m\Pi_a(\sigma)=E_a(\sigma)^m\Pi_a(\sigma)\) and linearity.

By the non-resonance assumption \eqref{eq:generic-sign-sum-nondegeneracy},
the numbers \(E_a(\sigma)\) are pairwise distinct.  Hence the Lagrange
polynomial \eqref{eq:Z-interpolation-polynomial} is well-defined and satisfies
\begin{equation}
z_{a,k}\bigl(E_a(\sigma)\bigr)=\sigma_k
\qquad
\text{for all }\sigma\in\{\pm1\}^{\cS}.
\end{equation}
Therefore
\begin{equation}
z_{a,k}(H_a)
=
\sum_{\sigma\in\{\pm1\}^{\cS}}
z_{a,k}\bigl(E_a(\sigma)\bigr)\Pi_a(\sigma)
=
\sum_{\sigma\in\{\pm1\}^{\cS}}
\sigma_k\Pi_a(\sigma)
=
\cZ_{a,k}.
\end{equation}
This proves \eqref{eq:Z-as-polynomial-in-H}.
\end{proof}

\begin{corollary}[Cross-commutativity of the mode occupations]
\label{cor:cross-commutativity-mode-occupations}
In the global-\(\lambda\) family, and on the common non-resonant locus
\eqref{eq:generic-sign-sum-nondegeneracy} for \(a=1,2\), the mode occupation
operators of the two hidden free-fermion structures commute:
\begin{equation}
[\cZ_{1,i},\cZ_{2,j}]=0,
\qquad
i,j=1,\dots,\cS .
\label{eq:cross-Z-commutativity}
\end{equation}
\end{corollary}

\begin{proof}
By Theorem~\ref{thm:mutual-commutativity}, in the global-\(\lambda\) family we
have
\begin{equation}
[H_1,H_2]=0 .
\end{equation}
Using Proposition~\ref{prop:mode-occupation-polynomial-in-H}, we may write
\begin{equation}
\cZ_{1,i}=z_{1,i}(H_1),
\qquad
\cZ_{2,j}=z_{2,j}(H_2).
\end{equation}
But polynomials in commuting operators commute.  Hence
\begin{equation}
[\cZ_{1,i},\cZ_{2,j}]
=
[z_{1,i}(H_1),z_{2,j}(H_2)]
=0 .
\end{equation}
\end{proof}

\begin{corollary}[Free fermions in the linear combinations]
  The linear combinations $H_1\pm H_2$ are free-fermionic, and their
  spectrum is given by
  \begin{equation}
        \label{eq:energies}
        E(\sigma,\tau)
        =
        \sum_{r=1}^{\cS}\sigma_r\varepsilon_{1,r}
        \pm 
        \sum_{r=1}^{\cS}\tau_r\varepsilon_{2,r},
        \qquad
        \sigma_r,\tau_r\in\{\pm1\},
        \qquad
        \cS=\left\lfloor \frac{M+1}{2}\right\rfloor,
  \end{equation}
where the single-particle energies $\varepsilon_{a,r}$ with $a=1,2$ are computed from the roots of the polynomials
$P_{a,M}(x)$. 
\end{corollary}

At this stage we can conveniently prove the following:

\begin{theorem}[Cross-commutativity of the transfer matrices]
\label{thm:cross-transfer-commutativity-new}
For all spectral parameters $u,v$,
\begin{equation}
[T_{1,M}(u),T_{2,M}(v)]=0.
\label{eq:cross-transfer-commutes}
\end{equation}
Consequently all higher charges generated by the two transfer matrices
commute across the two families.
\end{theorem}

  \begin{proof}
The single-family free-fermion diagonalization gives \cite{ffd,fermions-behind-the-disguise}
\[
T_{a,M}(u)
=
\prod_{k=1}^{\cS}
\left(\I-u\eps_{a,k}\cZ_{a,k}\right),
\qquad
\cZ_{a,k}=[\Psi_{a,k},\Psi_{a,-k}].
\]
The occupation involutions commute within each family. Their cross-commutativity was proved in
Corollary~\ref{cor:cross-commutativity-mode-occupations}.  
Therefore every factor in \(T_{1,M}(u)\) commutes with every factor in
\(T_{2,M}(v)\), proving the theorem.
\end{proof}

\subsection{Compatible edge operators and Dirac fermions}
\label{subsec:compatible-edges}

To have compatible fermionic operators, the edge operators must
also be compatible with the opposite transfer matrix.  Algebraically, we require
\begin{equation}
[\chi_1,\tA_j]=[\chi_1,\tB_j]=0,
\qquad
[\chi_2,A_j]=[\chi_2,B_j]=0
\qquad \text{for all }j.
\label{eq:cross-compatible-edge-algebra}
\end{equation}
Then
\begin{equation}
[\chi_1,T_{2,m}(v)]=0,
\qquad
[\chi_2,T_{1,m}(u)]=0
\end{equation}
for every $0\leq m\leq M$ and all $u,v$.

The two simplicial cliques we chose earlier are
\begin{equation}
  K_1=\{A_1,B_1\},\qquad K_2=\{\tA_{M-1},\tB_M\}.
\end{equation}
Recall that
\begin{equation}
 A_j=\ii\lambda_j B_jB_{j+1}\tB_j,
\qquad
\tA_j=-\ii\tlambda_j\tB_j\tB_{j+1}B_{j+1},
\qquad (j=1,\dots,M-1).
\end{equation}
This means that if
\begin{equation}
  \begin{split}
    \{\chi_1,B_1\}&=0,\qquad [\chi_1,B_j]=0\quad j>1,\\
    \{\chi_2,\tB_M\}&=0,\qquad [\chi_2,\tB_j]=0\quad j<M,\\
  \end{split}
\end{equation}
and if the $A_j$ and $\tA_j$ are given by the cubic products above while the cross-commutation relations
\eqref{eq:cross-compatible-edge-algebra} are also satisfied, 
then the edge operators are compatible with the simplicial cliques chosen above. 

There are two useful choices for the mutual relation of the two edge operators:
\begin{equation}
[\chi_1,\chi_2]=0
\label{eq:commuting-edge-algebra}
\end{equation}
or
\begin{equation}
\{\chi_1,\chi_2\}=0.
\label{eq:anti-commuting-edge-algebra}
\end{equation}
In these cases we obtain simple cross-relations between the fermionic operators. Indeed, 
we can show using \eqref{eq:Dirac-transfer-formula-new},
Theorem~\ref{thm:cross-transfer-commutativity-new}, and the edge-transfer
relations above, that the commutation relations of the two hidden Dirac families are
given by
\begin{equation}
\Psi_{1,k}\Psi_{2,\ell}
=
\eta\,\Psi_{2,\ell}\Psi_{1,k},
\qquad
\eta=\begin{cases}
+1,&[\chi_1,\chi_2]=0,\\
-1,&\{\chi_1,\chi_2\}=0.
\end{cases}
\label{eq:mutual-Dirac-statistics}
\end{equation}

We now consider the Ising and XY representations and find concrete edge operators satisfying all the necessary
conditions. For simplicity we focus on the case of \eqref{eq:commuting-edge-algebra}, which leads to commuting families
of fermionic ladder operators. In this case the edge operators can be chosen as completely localized operators in both
representations. It is also possible to find appropriate $\chi_{1,2}$ in the anti-commuting case, but then one of them
needs to be a non-local operator product.

In the Ising case we denote by $n=\lceil (M+1)/2 \rceil$ the number of spins on one chain. For the edge operators we
can choose for example 
\begin{equation}
  \chi_1=\sigma_1^x,\qquad
  \chi_2=
	  \begin{cases}
	    \tau^x_n \quad \text{$M$ odd}\\
	   \tau^z_n \quad \text{$M$ even}\\
	  \end{cases}.
\end{equation}

In the XY case we can choose
\begin{equation}
  \chi_1=X_1,\qquad \chi_2=X_{M+1}.
\end{equation}
These choices satisfy all requirements.

Finally, we remark that if we had chosen $K_1$ as the simplicial clique for both $H_1$ and $H_2$, then we would not be
able to satisfy all the necessary conditions for the edge operators, irrespective of the representation. 
A putative edge operator \(\chi_2\) would need to obey
\[
\{\chi_2,\widetilde B_1\}=0,
\qquad
\{\chi_2,\widetilde A_1\}=0,
\]
and cross-compatibility with \(H_1\) would require
\[
[\chi_2,B_j]=[\chi_2,A_j]=0
\qquad
\text{for all }j.
\]
However,
\[
A_1=\mathrm i\lambda_1 B_1B_2\widetilde B_1.
\]
Thus \([\chi_2,B_1]=[\chi_2,B_2]=0\), together with
\(\{\chi_2,\widetilde B_1\}=0\), forces
\[
\{\chi_2,A_1\}=0,
\]
contradicting the required condition \([\chi_2,A_1]=0\).  Therefore the
same-left-boundary choice cannot give cross-compatible edge operators.

\section{The degeneracies in the spectrum}

\label{sec:degs}

We are now in a position to discuss our main claims about the spectrum and its degeneracies.
It follows from the previous sections that the energy eigenvalues are of the form \eqref{eq:energies}.
We now discuss the degeneracies of these energy values, both in the Ising and XY representations.

Assume that the coupling constants $\beta_j$, $\tilde\beta_j$, $j=1,\dots,M$ are generic numbers, and in
particular, the two sets are not related to each other in any way.
Fix a finite $M$ and tune the perturbation
strength $\lambda$ continuously from 0 to some finite value.
Switching on an infinitesimal $\lambda$ can split existing
degeneracies, but it cannot increase them, except at isolated points in parameter space. 
At the same time, the energy formula above is rigid; in particular, the number of fermionic eigenmodes $\cS$ is the same
for $\lambda=0$ and for a non-zero $\lambda$. Only the
one-particle energies change as a function of $\lambda$. It follows that the degeneracies are exactly the same for
$\lambda=0$ and a finite non-zero $\lambda$, at least in a small neighbourhood of $\lambda=0$.

The actual degeneracies can be computed using simple counting arguments:
\begin{itemize}
\item For $M=2\ell+1$ the number of qubits used in the two representations is the same, given by $N=2\ell+2$. The number of
  fermionic eigenmodes is $\cS=\ell+1$ for each Hamiltonian $H_1$ and $H_2$; therefore the total number of eigenmodes
  coincides with the number of qubits. This implies that every energy level is non-degenerate.
\item For $M=2\ell$ the number of qubits used in the two representations is not the same:
    we have $N=2\ell+2$ and $N=2\ell+1$ qubits in the Ising and XY representations, respectively. The total number of
    eigenmodes is $2\cS=2\ell$. This implies that every energy level is degenerate with a homogeneous degeneracy of 4
    and 2 in the Ising and XY representations, respectively.
\end{itemize}

In Appendix \ref{app:numerics} we present concrete numerical evidence for the free-fermionic spectrum and these
degeneracies, in the case of the XY representation.

If we impose more symmetry on the combinations $H_1\pm H_2$, then the degeneracy pattern can change substantially. For
example, in the homogeneous case $\beta_j=\tilde \beta_j=\beta$, the two Hamiltonians have identical single-particle energies.
As a result, the
degeneracies of the energy eigenvalues will not be homogeneous. In particular, the null spaces of $H_1\pm H_2$ will be
exponentially degenerate.

\section{The homogeneous chain}
\label{sec:homogeneous-dispersion}

Here we treat the homogeneous chain and derive the dispersion relations and the standing-wave equations for the
quasi-momenta of the fermionic excitations. We show that these quantities also interpolate between the corresponding
quantities of the XY and FFD models.

We use the normalization
\begin{equation}
\beta_j=\tbeta_j=1
\qquad (j=1,\dots,M),
\qquad
\lambda_j=\tlambda_j=\lambda
\qquad (j=1,\dots,M-1).
\label{eq:homogeneous-beta-one-normalization}
\end{equation}
Thus, for both families,
\begin{equation}
B_j^2=\tB_j^2=\I,
\qquad
A_j^2=\tA_j^2=\lambda^2\I.
\label{eq:homogeneous-beta-one-squares}
\end{equation}
We assume $\lambda>0$ in the dispersion formulas.  The scalar spectral data
only depend on $\lambda^2$, so changing the sign of $\lambda$ does not change
the one-particle energies.

The two families have the same weighted independence polynomial, and the single-particle energies will be identical.

For either
family let $N_A(S)$ be the number of long generators in the independent set
$S$; that is, the number of $A_j$'s in the untilded family or $\tA_j$'s in the
tilded family.  Then
\begin{equation}
P_M^{(\lambda)}(x)
=
\sum_{S\in\cI_M}(-x)^{|S|}\lambda^{2N_A(S)}.
\label{eq:P-homogeneous-lambda-independence}
\end{equation}
If
\begin{equation}
P_M^{(\lambda)}(x)
=
\prod_{k=1}^{\cS}(1-\eps_k(\lambda)^2x),
\qquad
\cS=\left\lfloor\frac{M+1}{2}\right\rfloor,
\label{eq:P-homogeneous-lambda-factorization}
\end{equation}
then $x_k=\eps_k(\lambda)^{-2}$ are the roots of
$P_M^{(\lambda)}$.  In the present symmetric homogeneous case these are the
one-particle energies for both $H_1$ and $H_2$:
\begin{equation}
\eps_{1,k}=\eps_{2,k}=\eps_k(\lambda).
\end{equation}
Consequently the scalar part of the spectrum of $H_1\pm H_2$ is obtained from
two identical copies of this set of one-particle energies, as in
\eqref{eq:energies}.

\begin{corollary}[Polynomial recursion]
\label{cor:homogeneous-polynomial-recursion}
The homogeneous polynomial in the normalization
\eqref{eq:homogeneous-beta-one-normalization} satisfies
\begin{equation}
P_M^{(\lambda)}(x)
=
P_{M-1}^{(\lambda)}(x)
-xP_{M-2}^{(\lambda)}(x)
-\lambda^2xP_{M-3}^{(\lambda)}(x),
\qquad (M\geq3),
\label{eq:P-homogeneous-lambda-recursion}
\end{equation}
with
\begin{equation}
P_0^{(\lambda)}(x)=1,
\qquad
P_1^{(\lambda)}(x)=1-x,
\qquad
P_2^{(\lambda)}(x)=1-(2+\lambda^2)x.
\label{eq:P-homogeneous-lambda-initial}
\end{equation}
The generating function is
\begin{equation}
\sum_{M\geq0}P_M^{(\lambda)}(x)t^M
=
\frac{1-xt-\lambda^2xt^2}{1-t+xt^2+\lambda^2xt^3}.
\label{eq:P-homogeneous-lambda-generating}
\end{equation}
\end{corollary}

\begin{proof}
Decompose an independent set by the largest right endpoint.  It either uses no
generator touching the right boundary, uses the short generator $B_M$ or
$\tB_M$, or uses the long generator $A_{M-1}$ or $\tA_{M-1}$.  The short
weight is $1$ and the long weight is $\lambda^2$, giving the three terms in
\eqref{eq:P-homogeneous-lambda-recursion}.  The initial values are the direct
independent-set sums for algebra sizes $0,1,2$.  The generating function
follows by summing the recursion with these initial values.
\end{proof}

Equivalently,
\begin{equation}
P_M^{(\lambda)}(x)=
\sum_{\substack{r,s\geq0\\2r+3s\leq M+1}}
(-x)^{r+s}\lambda^{2s}
\binom{M+1-r-2s}{r+s}\binom{r+s}{r},
\label{eq:P-homogeneous-lambda-closed-form}
\end{equation}
where $r$ is the number of short generators and $s$ is the number of long
generators.

\subsection{Bulk dispersion and finite-size standing waves}
\label{subsec:bulk-dispersion}

For fixed $x$, the characteristic equation of
\eqref{eq:P-homogeneous-lambda-recursion} is
\begin{equation}
\mu^3-\mu^2+x\mu+\lambda^2x=0.
\label{eq:charpoly-lambda}
\end{equation}
For the physical branch with $x>0$, write the complex conjugate pair and the
real root as
\begin{equation}
\mu_\pm=B_\lambda\e^{\pm\ii p},
\qquad
\mu_0=1-U_\lambda,
\qquad
U_\lambda=2B_\lambda\cos p,
\qquad
0<p<\frac{\pi}{2}.
\label{eq:mu-param-lambda}
\end{equation}
Vieta's relations give
\begin{equation}
U_\lambda^2+\bigl(4\lambda^2\cos^2p-1-\lambda^2\bigr)U_\lambda
-4\lambda^2\cos^2p=0.
\label{eq:U-quadratic-lambda}
\end{equation}
The physical solution is
\begin{equation}
U_\lambda(p)=
\frac{
1+\lambda^2-4\lambda^2\cos^2p+
\sqrt{\bigl(1+\lambda^2-4\lambda^2\cos^2p\bigr)^2
+16\lambda^2\cos^2p}
}{2}.
\label{eq:U-physical-lambda}
\end{equation}
Then
\begin{equation}
B_\lambda(p)=\frac{U_\lambda(p)}{2\cos p},
\qquad
x_\lambda(p)=
\frac{U_\lambda(p)^2\bigl(U_\lambda(p)-1\bigr)}
{4\lambda^2\cos^2p},
\label{eq:x-dispersion-lambda}
\end{equation}
and the one-particle dispersion is
\begin{equation}
\eps_\lambda(p)=\frac{1}{\sqrt{x_\lambda(p)}}
=
\frac{2\lambda\cos p}
{U_\lambda(p)\sqrt{U_\lambda(p)-1}}.
\label{eq:epsilon-dispersion-lambda}
\end{equation}

The exact finite-size formula follows from partial fractions.  If the three
roots of \eqref{eq:charpoly-lambda} are denoted by
$\mu_+$, $\mu_-$, and $\mu_0$, then
\begin{equation}
P_M^{(\lambda)}(x)=
\sum_{a\in\{+,-,0\}}
\frac{\mu_a^{M+3}}{\prod_{b\neq a}(\mu_a-\mu_b)}.
\label{eq:partial-fraction-lambda}
\end{equation}
The decomposition is valid if the roots are all distinct, which is justified for almost all values of $p$.

Define
\begin{equation}
z_\lambda(p):=\mu_+-\mu_0
=B_\lambda(p)\e^{\ii p}-\bigl(1-U_\lambda(p)\bigr)
=\rho_\lambda(p)\e^{\ii\delta_\lambda(p)}.
\label{eq:z-rho-delta-lambda}
\end{equation}
Then
\begin{equation}
\rho_\lambda(p)=
\sqrt{\bigl(3B_\lambda(p)\cos p-1\bigr)^2
+B_\lambda(p)^2\sin^2p},
\qquad
\delta_\lambda(p)=
\arctan\frac{U_\lambda(p)\tan p}{3U_\lambda(p)-2}.
\label{eq:rho-delta-lambda}
\end{equation}
A direct simplification gives
\begin{equation}
P_M^{(\lambda)}(x_\lambda(p))=
\frac{B_\lambda(p)^{M+2}}{\rho_\lambda(p)\sin p}
\sin\bigl((M+3)p-\delta_\lambda(p)\bigr)
+
\frac{\bigl(1-U_\lambda(p)\bigr)^{M+3}}
{\rho_\lambda(p)^2}.
\label{eq:P-exact-phase-lambda}
\end{equation}
Thus the exact root condition is
\begin{equation}
\sin\bigl((M+3)p_k-\delta_\lambda(p_k)\bigr)
=
-\frac{\bigl(1-U_\lambda(p_k)\bigr)^{M+3}\sin p_k}
{\rho_\lambda(p_k)B_\lambda(p_k)^{M+2}}.
\label{eq:exact-quantization-lambda}
\end{equation}
Since $|1-U_\lambda(p)|<B_\lambda(p)$ on the physical branch, the right-hand
side is exponentially small in $M$.  The large-$M$ standing-wave condition is
\begin{equation}
(M+3)p_k-\delta_\lambda(p_k)
=\pi k+O\!\left(
\left|\frac{1-U_\lambda(p_k)}{B_\lambda(p_k)}\right|^{M+2}
\right),
\qquad
k=1,\dots,\cS.
\label{eq:asymptotic-quantization-lambda}
\end{equation}

\subsection{Spectral edge}
\label{subsec:edge}

At fixed finite $\lambda>0$ and $p\to\pi/2$,
\begin{equation}
\eps_\lambda(p)
\sim
\frac{2}{1+\lambda^2}\left(\frac{\pi}{2}-p\right),
\label{eq:linear-edge-lambda}
\end{equation}
so the smallest mode scales as $M^{-1}$.

\subsection{The $\lambda\to0$ limit}
\label{subsec:small-lambda-limit}

The endpoint $\lambda=0$ is obtained as a regular degeneration of the
finite-$\lambda$ formulas.  First, taking the limit in the closed form
\eqref{eq:P-homogeneous-lambda-closed-form} suppresses all terms containing
long generators:
\begin{equation}
P_M^{(0)}(x)
:=
\lim_{\lambda\to0}P_M^{(\lambda)}(x)
=
\sum_{\substack{r\geq0\\2r\leq M+1}}
(-x)^r\binom{M+1-r}{r}.
\label{eq:lambda0-closed-form}
\end{equation}
Equivalently, the same limit in the recursion and the generating function gives
\begin{equation}
P_M^{(0)}(x)=P_{M-1}^{(0)}(x)-xP_{M-2}^{(0)}(x),
\qquad (M\geq2),
\qquad
P_0^{(0)}=1,
\qquad
P_1^{(0)}=1-x,
\label{eq:lambda0-recursion}
\end{equation}
and
\begin{equation}
\sum_{M\geq0}P_M^{(0)}(x)t^M
=
\frac{1-xt}{1-t+xt^2}.
\label{eq:lambda0-generating}
\end{equation}

We now take the limit directly in the finite-$\lambda$ dispersion formulas.
For fixed $0<p<\pi/2$, write $c=\cos p$.  From
\eqref{eq:U-physical-lambda} one obtains
\begin{equation}
U_\lambda(p)
=
1+\lambda^2-4\lambda^4c^2+O(\lambda^6).
\label{eq:lambda0-U-expansion}
\end{equation}
Substitution into \eqref{eq:x-dispersion-lambda} gives
\begin{equation}
x_\lambda(p)
=
\frac{1}{4c^2}
\left[
1+2(1-2c^2)\lambda^2+O(\lambda^4)
\right],
\qquad
\eps_\lambda(p)
=
2c
\left[
1+(2c^2-1)\lambda^2+O(\lambda^4)
\right].
\label{eq:lambda0-dispersion-limit}
\end{equation}
Thus the limiting parametrization is
\begin{equation}
x_0(p)=\frac{1}{4\cos^2p},
\qquad
\eps_0(p)=2\cos p.
\label{eq:lambda0-x-epsilon}
\end{equation}

The finite-size standing-wave condition also has a direct limit.  From
\eqref{eq:rho-delta-lambda},
\begin{equation}
B_\lambda(p)\longrightarrow \frac{1}{2\cos p},
\qquad
\rho_\lambda(p)\longrightarrow \frac{1}{2\cos p},
\qquad
\delta_\lambda(p)\longrightarrow p.
\label{eq:lambda0-phase-limits}
\end{equation}
Moreover,
\begin{equation}
1-U_\lambda(p)=-\lambda^2+O(\lambda^4),
\end{equation}
so the second term in \eqref{eq:P-exact-phase-lambda} vanishes as
$\lambda\to0$.  Taking the limit of \eqref{eq:P-exact-phase-lambda} therefore
gives
\begin{equation}
P_M^{(0)}(x_0(p))
=
\frac{\sin((M+2)p)}{(2\cos p)^{M+1}\sin p}.
\label{eq:lambda0-sine-polynomial}
\end{equation}
The exact root condition follows immediately:
\begin{equation}
\sin((M+2)p_k)=0.
\end{equation}
Hence
\begin{equation}
(M+2)p_k=\pi k,
\qquad
\eps_k(0)=2\cos\frac{\pi k}{M+2},
\qquad
k=1,\dots,\cS.
\label{eq:lambda0-spectrum}
\end{equation}
This is the open-path one-particle quantization obtained as the
$\lambda\to0$ limit of the general standing-wave equation.

The same limiting calculation also gives the first correction to the momenta.
For fixed $M$, the right-hand side of
\eqref{eq:exact-quantization-lambda} is $O(\lambda^{2M+6})$, while
\begin{equation}
\delta_\lambda(p)
=
p-\lambda^2\sin(2p)+O(\lambda^4).
\label{eq:lambda0-delta-expansion}
\end{equation}
Thus
\begin{equation}
(M+2)p_k+\lambda^2\sin(2p_k)
=
\pi k+O(\lambda^4),
\end{equation}
or, with $p_k^{(0)}=\pi k/(M+2)$,
\begin{equation}
p_k(\lambda)
=
p_k^{(0)}
-\frac{\lambda^2}{M+2}\sin(2p_k^{(0)})
+O(\lambda^4).
\label{eq:lambda0-momentum-correction}
\end{equation}

\subsection{The $\lambda\to\infty$ limit}
\label{subsec:large-lambda-limit}

The FFD endpoint is obtained from the finite-$\lambda$ formulas after the
natural rescaling
\begin{equation}
x=\frac{y}{\lambda^2},
\qquad
\widehat\eps=\frac{\eps}{\lambda},
\qquad
y=\frac{1}{\widehat\eps^2}.
\label{eq:large-lambda-scaling}
\end{equation}
Substituting $x=y/\lambda^2$ into the closed form
\eqref{eq:P-homogeneous-lambda-closed-form} gives
\begin{equation}
P_M^{(\lambda)}\!\left(\frac{y}{\lambda^2}\right)
=
\sum_{\substack{r,s\geq0\\2r+3s\leq M+1}}
(-y)^{r+s}\lambda^{-2r}
\binom{M+1-r-2s}{r+s}\binom{r+s}{r}.
\label{eq:large-lambda-scaled-closed-form}
\end{equation}
Therefore only the terms with $r=0$ survive, and we obtain the polynomial of the original FFD model \cite{ffd}
\begin{equation}
R_M(y)
:=
\lim_{\lambda\to\infty}
P_M^{(\lambda)}\!\left(\frac{y}{\lambda^2}\right)
=
\sum_{\substack{s\geq0\\3s\leq M+1}}
(-y)^s\binom{M+1-2s}{s}.
\label{eq:large-lambda-R-closed-form}
\end{equation}
Equivalently, taking the same limit in the recursion and in the generating
function gives
\begin{equation}
R_M(y)=R_{M-1}(y)-yR_{M-3}(y),
\qquad
R_0=1,
\qquad
R_1=1,
\qquad
R_2=1-y,
\label{eq:large-lambda-R-recursion}
\end{equation}
and
\begin{equation}
\sum_{M\geq0}R_M(y)t^M
=
\frac{1-yt^2}{1-t+yt^3}.
\label{eq:large-lambda-R-generating}
\end{equation}
Thus the finite roots of $R_M$ describe precisely the modes for which
$x=O(\lambda^{-2})$, or equivalently $\eps=O(\lambda)$.  The degree of
$R_M$ is $\lfloor(M+1)/3\rfloor$; the remaining roots of
$P_M^{(\lambda)}$ are sent to $y=\infty$ and therefore have
$\widehat\eps\to0$ in this scaling.

We now obtain the limiting dispersion by taking $\lambda\to\infty$ in the
finite-$\lambda$ branch.  For the linearly growing modes, the momentum remains
in the interval $0<p<\pi/3$; we denote it by $\theta$.  Put
\[
c=\cos\theta,
\qquad
q=4c^2-1>0.
\]
The large-$\lambda$ expansion of \eqref{eq:U-physical-lambda} is
\begin{equation}
U_\lambda(\theta)
=
\frac{4c^2}{q}
-\frac{4c^2}{q^3}\lambda^{-2}
+O(\lambda^{-4}).
\label{eq:large-lambda-U-expansion}
\end{equation}
Consequently
\begin{equation}
U_\infty(\theta)=\frac{4\cos^2\theta}{4\cos^2\theta-1},
\qquad
B_\infty(\theta)
=
\frac{2\cos\theta}{4\cos^2\theta-1}
=
\frac{\sin2\theta}{\sin3\theta}.
\label{eq:large-lambda-U-B}
\end{equation}
Taking the same limit in \eqref{eq:x-dispersion-lambda} gives
\begin{equation}
\lambda^2 x_\lambda(\theta)
=
\frac{4c^2}{q^3}
\left[
1-\frac{2(1+2c^2)}{q^2}\lambda^{-2}
+O(\lambda^{-4})
\right].
\label{eq:large-lambda-x-expansion}
\end{equation}
Hence
\begin{equation}
y_\infty(\theta)
=
\frac{4\cos^2\theta}{(4\cos^2\theta-1)^3},
\qquad
\widehat\eps_\infty(\theta)
=
\frac{1}{\sqrt{y_\infty(\theta)}}
=
\frac{(4\cos^2\theta-1)^{3/2}}{2\cos\theta},
\qquad
0<\theta<\frac{\pi}{3}.
\label{eq:large-lambda-dispersion}
\end{equation}
Equivalently,
\begin{equation}
\frac{\eps_\lambda(\theta)}{\lambda}
\longrightarrow
\frac{(4\cos^2\theta-1)^{3/2}}{2\cos\theta}.
\label{eq:large-lambda-scaled-epsilon-limit}
\end{equation}
This is Fendley's $q=0$ FFD dispersion after identifying his momentum as
$p_{\rm Fen}=3\theta$ \cite{ffd}.  The first correction follows from
\eqref{eq:large-lambda-x-expansion}:
\begin{equation}
\frac{\eps_\lambda(\theta)}{\lambda}
=
\frac{(4\cos^2\theta-1)^{3/2}}{2\cos\theta}
\left[
1+
\frac{1+2\cos^2\theta}{(4\cos^2\theta-1)^2}\lambda^{-2}
+O(\lambda^{-4})
\right].
\label{eq:large-lambda-epsilon-correction}
\end{equation}

The limit is non-uniform at $\theta=\pi/3$.  For fixed
$p\in(\pi/3,\pi/2)$, put $d=1-4\cos^2p>0$.  The same finite-$\lambda$
formula gives
\begin{equation}
U_\lambda(p)
=
\lambda^2 d+d^{-1}+O(\lambda^{-2}),
\end{equation}
and therefore
\begin{equation}
\eps_\lambda(p)
=
\frac{2\cos p}{\lambda^2(1-4\cos^2p)^{3/2}}
\left(1+O(\lambda^{-2})\right).
\label{eq:large-lambda-collapsing-branch}
\end{equation}
This part of the physical branch collapses to zero energy in the scaled
Hamiltonian $H/\lambda$.  At the boundary $p=\pi/3$ one instead finds
\begin{equation}
\eps_\lambda(\pi/3)\sim \lambda^{-1/2}.
\end{equation}

It remains to take the same limit in the finite-size standing-wave formula.
From \eqref{eq:rho-delta-lambda} we obtain
\begin{equation}
\rho_\infty(\theta)
=
\frac{\sqrt{8\cos^2\theta+1}}{4\cos^2\theta-1},
\qquad
\delta_\infty(\theta)
=
\arctan
\frac{2\cos\theta\sin\theta}{2\cos^2\theta+1}.
\label{eq:large-lambda-rho-delta}
\end{equation}
Taking $\lambda\to\infty$ in \eqref{eq:P-exact-phase-lambda} gives the exact
limiting phase formula
\begin{equation}
R_M(y_\infty(\theta))
=
\frac{B_\infty(\theta)^{M+2}}{\rho_\infty(\theta)\sin\theta}
\sin\bigl((M+3)\theta-\delta_\infty(\theta)\bigr)
+
\frac{\bigl(1-U_\infty(\theta)\bigr)^{M+3}}
{\rho_\infty(\theta)^2}.
\label{eq:large-lambda-R-phase}
\end{equation}
Therefore the exact limiting root condition is
\begin{equation}
\sin\bigl((M+3)\theta_\ell-\delta_\infty(\theta_\ell)\bigr)
=
-\frac{
\bigl(1-U_\infty(\theta_\ell)\bigr)^{M+3}\sin\theta_\ell
}{
\rho_\infty(\theta_\ell)B_\infty(\theta_\ell)^{M+2}
},
\qquad
\ell=1,\dots,\left\lfloor\frac{M+1}{3}\right\rfloor.
\label{eq:large-lambda-exact-quantization}
\end{equation}
Since
\begin{equation}
\left|
\frac{1-U_\infty(\theta)}{B_\infty(\theta)}
\right|
=
\frac{1}{2\cos\theta}
<1
\qquad
(0<\theta<\pi/3),
\end{equation}
the large-$M$ standing-wave condition is
\begin{equation}
(M+3)\theta_\ell-\delta_\infty(\theta_\ell)
=
\pi\ell
+
O\!\left((2\cos\theta_\ell)^{-(M+2)}\right),
\qquad
\ell=1,\dots,\left\lfloor\frac{M+1}{3}\right\rfloor.
\label{eq:large-lambda-asymptotic-quantization}
\end{equation}

\section{Conclusions}
\label{sec:concl}

In this work we presented a spin-chain model that falls under the broad umbrella of ``free fermions in disguise'' and
is free from exponential degeneracies. In Subsection \ref{sec:origin} we explained why it is difficult to find such a model.
We now summarize the mechanism that makes these special properties possible.

In our case the final Hamiltonian is a combination of two Hamiltonians $H_1$ and $H_2$. Both have 
$2M-1$ terms, where $M$ is the size of a fundamental graph-Clifford algebra.
Thus, we have a total of $4M-2$ generators. However, they are
not independent, and there are $2M-2$ independent algebraic relations between them.
As a result, the combined generator set can be represented on $M+\ordo(1)$ qubits.
At the same time, the
frustration graphs of the two Hamiltonians $H_1$ and $H_2$ allow for $\lfloor (M+1)/2\rfloor$ fermionic eigenmodes, leading
to a total of $2\lfloor (M+1)/2\rfloor$ eigenmodes. Thus, we have enough eigenmodes to avoid exponential
degeneracies. 

However, there is an additional crucial point: even though $H_1$ and $H_2$ have ECF-type frustration graphs, this is not
true for the combined generator set. Consequently, the model would not even be integrable for arbitrary coupling constants
in the combined generator set. If we choose the couplings such that special quadratic relations are satisfied,
then we obtain the necessary cancellations that ensure commutativity of $H_1$ and $H_2$ and the compatibility of their
free-fermionic structures. Such relations were used in the work \cite{sajat-FP-model} for a single Hamiltonian. In our case we can
find the relations for the combined generator set of two Hamiltonians $H_1$ and $H_2$, thereby avoiding the
exponential degeneracies.

Our Hamiltonian $H_1-H_2$ in the XY representation is structurally appealing: it interpolates between the
$U(1)$-invariant XY model (given by the Dzyaloshinskii–Moriya interaction term) and the original FFD model of
Fendley. The degeneracies are constant along the line connecting 
the two models, except for the FFD endpoint, which corresponds to $\lambda=\infty$ in our conventions. For numerical
data about the interpolation, see Appendix \ref{app:degchange}.

Given that our model is a perturbation of a Jordan-Wigner solvable model, and that the free-fermionic
operators change continuously with $\lambda$, a natural question is whether there is any sort of ``simple'' connection
between the models. Is there perhaps a finite-depth quantum circuit or an MPO with finite bond dimension that could
connect the structures of the FFD-type perturbation and the XY chain? If such a circuit or MPO exists, it would suggest
that even the original FFD model can be obtained from the XY chain by such ``simple'' means. Currently it is not known
whether such a connection can or cannot exist. 

We also note an interesting observation: If we apply the standard techniques of integrable models to the transfer
matrices, then we can define appropriate Lax operators and also $R$-matrices that perform intertwining
steps. Interestingly, we find that the $R$-matrix of the model is exactly the same $9\times 9$ matrix for all
(inhomogeneous) coupling 
constants. This includes the FFD model itself as an endpoint of the interpolation. A simplification arises only in
the XY endpoint, where only a $4\times 4$ block of the $R$-matrix remains relevant. These results are not used in our
main derivations, but they are presented as complementary material in Appendix \ref{app:MPO}. Results about the
integrability structures may be relevant to the study of the models with periodic boundary conditions.

In future work we plan to investigate other mechanisms that lead to non-degenerate FFD-type models. As we were 
preparing this manuscript, we discovered another such model with a related but slightly different mechanism. We
leave this topic to further research.

\section*{Acknowledgements}

We thank Kohei Fukai and Istv\'an Vona for many useful discussions and collaborations on related
topics. We also thank Istv\'an Vona for an independent AI-based check of the results of this work.

The author was supported by the Hungarian National Research,
Development and Innovation Office, NKFIH Grant No. K-145904.

\appendix

\section{Numerical data}
\label{app:numerics}

\subsection{Free fermions for $H_1-H_2$}

We now give a direct numerical check of the spectral formulas in
the XY representation.  We restrict to the combination $H_1-H_2$ and choose
\[
        \lambda=1,\qquad
        \beta_k=\sin k,\qquad
        \widetilde\beta_k=\cos k ,
        \qquad k=1,\ldots,M ,
\]
where the argument of the trigonometric functions is measured in radians. For our purposes these couplings serve
as completely generic numbers.

For these parameters we diagonalize the Pauli-chain Hamiltonian
\[
\begin{aligned}
H_1-H_2
={}&
\sum_{j=1}^{M}
\left(
  \beta_j Y_jX_{j+1}
  -\widetilde\beta_j X_jY_{j+1}
\right)
\\
&+
\sum_{j=1}^{M-1}
\left(
  \alpha_j Z_jX_{j+1}X_{j+2}
  -\widetilde\alpha_j X_jX_{j+1}Z_{j+2}
\right),
\end{aligned}
\]
with
\[
        \alpha_j=\beta_j\beta_{j+1}\widetilde\beta_j,
        \qquad
        \widetilde\alpha_j
        =
        \widetilde\beta_j\widetilde\beta_{j+1}\beta_{j+1}.
\]
The chain length in the XY representation is
\[
        L=M+1 ,
\]
so the exact diagonalization is performed on a Hilbert space of dimension
$2^L$.

The energies obtained from the free-fermionic solution are compared with the
sign-sum expression \eqref{eq:energies}.

The common multiplicity in the XY representation is
\[
        2^{L-2\cS}.
\]

The one-particle energies are extracted from the spectral polynomials
\[
        P_{a,M}(x)=\prod_{r=1}^{\cS}
        \left(1-\varepsilon_{a,r}^{\,2}x\right),
        \qquad a=1,2.
\]
For the two tested system sizes the resulting polynomials are given in Table \ref{tab:xy-polynomials}.

\begin{table}[t]
\centering
\caption{Spectral polynomials for $M=4,5$.}
\label{tab:xy-polynomials}
\begin{tabular}{c c l}
\toprule
$M$ & Polynomial & $P_{a,M}(x)$ \\
\midrule
$4$ & $P_{1,4}$ &
$1-2.312499318432\,x+0.999015949013\,x^2$ \\
$4$ & $P_{2,4}$ &
$1-2.157454394006\,x+0.572702017908\,x^2$ \\
$5$ & $P_{1,5}$ &
$1-3.457052321338\,x+2.972333349787\,x^2
 -0.012966539921\,x^3$ \\
$5$ & $P_{2,5}$ &
$1-2.269530745604\,x+0.708647863841\,x^2
 -0.023021856212\,x^3$ \\
\bottomrule
\end{tabular}
\end{table}

The corresponding positive one-particle energies are listed in
Table~\ref{tab:xy-one-particle}.  These are the numbers entering the
sign-sum expression above.

\begin{table}[t]
\centering
\caption{Positive one-particle energies.}
\label{tab:xy-one-particle}
\begin{tabular}{c c c c c}
\toprule
$M$ & Sector & $\varepsilon_{a,1}$ & $\varepsilon_{a,2}$
& $\varepsilon_{a,3}$ \\
\midrule
$4$ & $a=1$ & $0.758261362983$ & $1.318157435149$ & -- \\
$4$ & $a=2$ & $0.556771270364$ & $1.359213061482$ & -- \\
$5$ & $a=1$ & $0.066217399570$ & $1.254816328025$
& $1.370439185172$ \\
$5$ & $a=2$ & $0.191676889058$ & $0.573733517849$
& $1.379717567585$ \\
\bottomrule
\end{tabular}
\end{table}

We next list the resulting exact-diagonalization spectrum of $H_1-H_2$.
Only positive energies are shown; the spectrum is symmetric under
$E\mapsto -E$.  For $M=4$, every listed positive energy and its negative
partner have multiplicity $2$, in agreement with $2^{L-2\cS}=2$.  For
$M=5$, all levels are simple, in agreement with $2^{L-2\cS}=1$.

\begin{table}[t]
\centering
\caption{Positive exact-diagonalization energies of $H_1-H_2$ for $M=4$.
Each listed energy and its negative partner have multiplicity $2$.}
\label{tab:xy-spectrum-M4}
\begin{tabular}{c c c c}
\toprule
\multicolumn{4}{c}{$M=4$} \\
\midrule
$0.160434466286$ & $0.242545718951$
& $1.273977007015$ & $1.356088259680$ \\
$1.362337863284$ & $2.475880404013$
& $2.878860589250$ & $3.992403129979$ \\
\bottomrule
\end{tabular}
\end{table}

\begin{table}[t]
\centering
\caption{Positive exact-diagonalization energies of $H_1-H_2$ for $M=5$.
All listed energies and their negative partners are simple.}
\label{tab:xy-spectrum-M5}
\begin{tabular}{c c c c}
\toprule
\multicolumn{4}{c}{$M=5$} \\
\midrule
$0.413910139136$ & $0.432466903961$
& $0.546344938275$ & $0.564901703101$ \\
$0.663712618255$ & $0.796147417394$
& $0.797263917252$ & $0.815820682077$ \\
$0.929698716391$ & $0.948255481217$
& $1.047066396370$ & $1.179501195510$ \\
$1.561377174834$ & $1.579933939660$
& $1.693811973973$ & $1.712368738799$ \\
$1.811179653953$ & $1.943614453092$
& $1.944730952950$ & $1.963287717775$ \\
$2.077165752089$ & $2.095722516915$
& $2.194533432068$ & $2.326968231208$ \\
$3.173345274305$ & $3.305780073444$
& $3.556699052421$ & $3.689133851560$ \\
$4.320812310003$ & $4.453247109142$
& $4.704166088119$ & $4.836600887258$ \\
\bottomrule
\end{tabular}
\end{table}

The agreement between exact diagonalization and the sign-sum construction
is at the level of numerical roundoff in both examples.
No accidental
degeneracies are observed in these two tests.  The only degeneracy for
$M=4$ is the uniform twofold multiplicity in the XY representation
$2^{L-2\cS}=2$, while for $M=5$ the full spectrum is non-degenerate.

\subsection{The change in the degeneracy pattern}


\label{app:degchange}

Here we display the change in the degeneracy pattern as we move from the XY point to the FFD point. To do this, we
choose $M=4$, random $\beta_k$ and $\tilde\beta_k$ couplings, and vary $\lambda$ from zero to infinity. To keep the
spectra finite, we renormalize the Hamiltonian and present data for the spectrum of
\begin{equation}
  H_\theta\equiv \cos(\theta)H_{0}+\sin(\theta) H_{FFD},
\end{equation}
where
\begin{equation}
  \begin{split}
    H_0&=\sum_{j=1}^M B_j-\tilde B_j,\qquad H_{FFD}=\sum_{j=1}^{M-1} A_j-\tilde A_j.
  \end{split}
\end{equation}
We performed exact diagonalization and present spectral data for $\theta\in [0,\pi/2]$. The random couplings are
shown in Table \ref{tab:ra}. For reference, we present concrete numerical values for the energies at $\theta=0$ and
$\theta=\pi/2$ in Table \ref{tab:re}. Finally, the spectrum is plotted in Fig. \ref{fig:Hth}.

\begin{figure}[t]
  \centering
\includegraphics[scale=0.4]{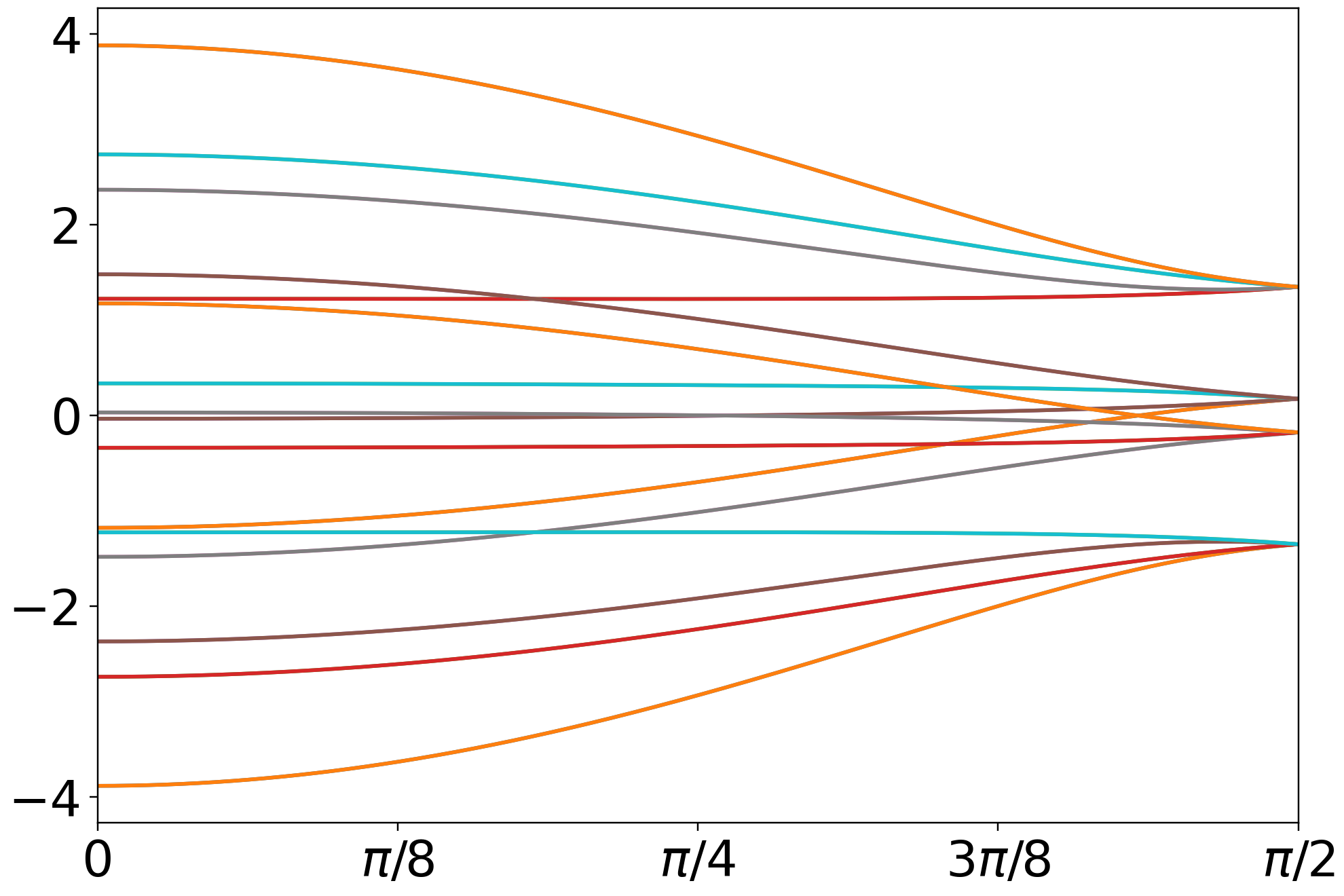}  
  \caption{The spectrum of $H_\theta$ as a function of $\theta$, for $M=4$. For each value of $\theta\ne \pi/2$ there are 16
    different eigenvalues, each with a homogeneous degeneracy of 2, 
    corresponding to the chain having length $L=M+1=5$. They collapse to 4 eigenvalues at $\theta=\pi/2$, corresponding
    to the FFD model, which has $\cS=2$ eigenmodes for $M=4$.}
  \label{fig:Hth}
\end{figure}

\begin{table}[h]
\centering
\caption{Random couplings used for the $M=4$ interpolation.
}
\begin{tabular}{c c c}
\toprule
$k$ & $\beta_k$ & $\widetilde\beta_k$ \\
\midrule
1 & $0.773956048556$ & $0.094177347888$ \\
2 & $0.438878439752$ & $0.975622351637$ \\
3 & $0.858597919911$ & $0.761139701990$ \\
4 & $0.697368029059$ & $0.786064305277$ \\
\bottomrule
\end{tabular}
\label{tab:ra}
\end{table}

\begin{table}[h]
\centering
\caption{Spectra of the Hamiltonians $H_\theta$
  at $\theta=0$ and $\theta=\pi/2$, for $M=4$.
Only positive energies are listed; the spectrum is symmetric under $E\mapsto -E$.
At $\theta=0$, each listed positive level and its negative partner are twofold degenerate.
At $\theta=\pi/2$, each listed positive level and its negative partner are eightfold degenerate.}
\begin{tabular}{c c}
\toprule
Endpoint & Positive energies $E>0$ \\
\midrule
$\theta=0$
&
\begin{tabular}{@{}c@{}}
$0.032944896137$ \\
$0.337604621603$ \\
$1.176669844980$ \\
$1.224894658396$ \\
$1.481329570446$ \\
$2.368619607239$ \\
$2.739169124980$ \\
$3.882894073823$
\end{tabular}
\\
\midrule
$\theta=\pi/2$
&
\begin{tabular}{@{}c@{}}
$0.176626910954$ \\
$1.349446176826$
\end{tabular}
\\
\bottomrule
\end{tabular}
\label{tab:re}
\end{table}

\section{Factorization of the transfer matrix}

\label{app:factor}

Here we prove that simple factorized formulas exist for the transfer matrices $T_{1,2}(u)$, analogous to the one
originally derived
in \cite{ffd}. 

Consider now a single Hamiltonian of the form
\begin{equation}
H=\sum_{j=1}^{M-1} A_j+\sum_{j=1}^{M} B_j,
\label{eq:H-factorization-generic}
\end{equation}
where the generators obey the interval anti-commutation pattern of one chain and
have scalar squares
\begin{equation}
A_j^2=\alpha_j^2\I,
\qquad
B_j^2=\beta_j^2\I.
\end{equation}
For non-zero weights set $\hA_j=A_j/\alpha_j$ and
$\hB_j=B_j/\beta_j$.  Its transfer matrix is
\begin{equation}
T_M(u)=\sum_{S\in\cI_M}(-u)^{|S|}
\prod_{A_j\in S} A_j
\prod_{B_k\in S} B_k.
\label{eq:T-factorization-generic}
\end{equation}
The order is immaterial inside each independent set, because the selected
generators commute.  The same right-boundary decomposition used in the main
text gives
\begin{equation}
T_m(u)=T_{m-1}(u)-uB_mT_{m-2}(u)-uA_{m-1}T_{m-3}(u),
\qquad m\geq2,
\label{eq:T-factorization-recursion}
\end{equation}
with
\begin{equation}
T_{-1}(u)=T_0(u)=\I,
\qquad
T_1(u)=\I-uB_1 .
\label{eq:T-factorization-initial}
\end{equation}
Indeed, an independent set for the prefix with short generators
$B_1,\ldots,B_m$ either contains no generator touching the right boundary,
contains $B_m$, or contains $A_{m-1}$.  In the last two cases the remaining
part lies in the prefixes of lengths $m-2$ and $m-3$, respectively.  For
$m=2$ this gives
\[
T_2(u)=\I-u(B_1+B_2+A_1),
\]
as required.

Define, for any involution $O^2=\I$,
\begin{equation}
g(O,\phi)=\cos\frac{\phi}{2}+O\sin\frac{\phi}{2}.
\label{eq:g-factorization}
\end{equation}
Choose the angles recursively by
\begin{align}
\phi_{-1}^A&=\phi_0^A=\phi_0^B=0,
&\sin\phi_1^B&=-u\,\beta_1,
\label{eq:factor-angles-init}
\\
\sin\phi_j^A&=
-\frac{u\,\alpha_j}{
\cos\phi_{j-2}^A\,
\cos\phi_{j-1}^B\,
\cos\phi_{j-1}^A\,
\cos\phi_j^B},
&1\le j&\le M-1,
\label{eq:factor-angles-A}
\\
\sin\phi_{j+1}^B&=
-\frac{u\,\beta_{j+1}}{
\cos\phi_j^A\,
\cos\phi_{j-1}^A\,
\cos\phi_j^B},
&1\le j&\le M-1.
\label{eq:factor-angles-B}
\end{align}
The angles are well-defined at least in a neighbourhood of $u=0$ and can then be continued analytically.

With
\begin{equation}
g_j^A=\cos\frac{\phi_j^A}{2}+\hA_j\sin\frac{\phi_j^A}{2},
\qquad
g_j^B=\cos\frac{\phi_j^B}{2}+\hB_j\sin\frac{\phi_j^B}{2},
\label{eq:gates-factorization}
\end{equation}
\begin{equation}
G_M=g_1^B g_1^A g_2^B g_2^A \cdots g_{M-1}^B g_{M-1}^A g_M^B,
\label{eq:GM-factorization}
\end{equation}
and
\begin{equation}
G_M^{\mathrm{rev}}=
 g_M^B g_{M-1}^A g_{M-1}^B \cdots g_1^A g_1^B,
\label{eq:GMrev-factorization}
\end{equation}
one obtains the exact factorization
\begin{equation}
T_M(u)=G_M(u)\,G_M^{\mathrm{rev}}(u).
\label{eq:factorization}
\end{equation}
The proof is an algebraic induction based on
\eqref{eq:T-factorization-recursion} and the elementary identities
\begin{equation}
g(O,\phi)^2=\I+O\sin\phi,
\qquad
g(O,\phi)Xg(O,\phi)=\begin{cases}
(\cos\phi)X,& \{O,X\}=0,\\[1mm]
X(\I+O\sin\phi),& [O,X]=0.
\end{cases}
\label{eq:g-basic-identities}
\end{equation}
They follow by writing $g(O,\phi)=c+sO$, with
$c=\cos(\phi/2)$ and $s=\sin(\phi/2)$.  Then
$(c+sO)^2=\I+2csO=\I+O\sin\phi$.  If $XO=-OX$, then
$(c+sO)X=X(c-sO)$, giving $g(O,\phi)Xg(O,\phi)=X(c^2-s^2)$.  If $XO=OX$,
then $X$ can be moved through both factors and the first identity applies.

For the induction, set
\begin{equation}
F_m(u)=G_m(u)G_m^{\mathrm{rev}}(u),
\qquad
F_{-1}(u)=F_0(u)=\I .
\label{eq:F-factorization-def}
\end{equation}
For $m=1$,
\[
F_1(u)=(g_1^B)^2
=\I+\hB_1\sin\phi_1^B
=\I-u\beta_1\hB_1
=\I-uB_1
=T_1(u).
\]
For $m\geq2$ one has
\[
G_m=G_{m-1}g_{m-1}^Ag_m^B,
\qquad
G_m^{\mathrm{rev}}=g_m^Bg_{m-1}^AG_{m-1}^{\mathrm{rev}}.
\]
Using \eqref{eq:g-basic-identities},
\begin{align}
F_m
&=G_{m-1}g_{m-1}^A(g_m^B)^2g_{m-1}^A
  G_{m-1}^{\mathrm{rev}}\nonumber\\
&=F_{m-1}
+\sin\phi_{m-1}^A\,G_{m-1}\hA_{m-1}G_{m-1}^{\mathrm{rev}}
\nonumber\\
&\hspace{2em}
+\sin\phi_m^B\cos\phi_{m-1}^A\,
G_{m-1}\hB_mG_{m-1}^{\mathrm{rev}} .
\label{eq:factor-app-1}
\end{align}
The operator $\hA_{m-1}$ anti-commutes exactly with
$g_{m-3}^A$, $g_{m-2}^B$, $g_{m-2}^A$, and $g_{m-1}^B$ inside $G_{m-1}$, with
the convention that gates with non-positive indices are absent.  It commutes
with all other gates.  Therefore
\begin{equation}
G_{m-1}\hA_{m-1}G_{m-1}^{\mathrm{rev}}
=
\cos\phi_{m-3}^A\cos\phi_{m-2}^B
\cos\phi_{m-2}^A\cos\phi_{m-1}^B\,
\hA_{m-1}F_{m-3}.
\label{eq:factor-app-2}
\end{equation}
Similarly, $\hB_m$ anti-commutes exactly with
$g_{m-2}^A$ and $g_{m-1}^B$ inside $G_{m-1}$, and hence
\begin{equation}
G_{m-1}\hB_mG_{m-1}^{\mathrm{rev}}
=\cos\phi_{m-2}^A\cos\phi_{m-1}^B\,\hB_mF_{m-2}.
\label{eq:factor-app-3}
\end{equation}
Substituting \eqref{eq:factor-app-2} and \eqref{eq:factor-app-3} into
\eqref{eq:factor-app-1}, and then using
\eqref{eq:factor-angles-A} and \eqref{eq:factor-angles-B}, gives
\begin{equation}
F_m(u)=F_{m-1}(u)-uA_{m-1}F_{m-3}(u)-uB_mF_{m-2}(u).
\label{eq:F-factorization-recursion}
\end{equation}
This is the recursion \eqref{eq:T-factorization-recursion}, with the same
initial data \eqref{eq:T-factorization-initial}.  Hence
$F_m(u)=T_m(u)$ for every $m$, and in particular
\eqref{eq:factorization} holds for $m=M$.

\section{MPO representations in the XY representation}

\label{app:MPO}

In this appendix we give a self-contained MPO construction for either the $H_1$ or the $H_2$ family
in the XY representation. These MPOs are not used in any of the proofs in the main text; instead we present them to
complement the material. We note that for the fermionic operators we find an MPO of bond dimension 4, and this is
a new result even for the FFD model of Fendley. The results below are valid for every spectral parameter $u$, and not
only the ``on-shell'' ones.

We use \(L\) for the length of the Pauli chain and
\(M=L-1\) for the number of \(B\)-generators.  For simplicity we focus only on $H_1$ and write
\begin{equation}
H=\sum_{j=1}^{M}B_j+\sum_{j=1}^{M-1}A_j,
\qquad
B_j=\beta_jY_jX_{j+1},
\qquad
A_j=\alpha_jZ_jX_{j+1}X_{j+2}.
\label{eq:mpo-single-chain}
\end{equation}
Here \(j=1,\dots,M\) for \(B_j\) and \(j=1,\dots,M-1\) for \(A_j\).
We use the boundary convention
\begin{equation}
\beta_j=0\quad (j\notin\{1,\dots,M\}),
\qquad
\alpha_j=0\quad (j\notin\{1,\dots,M-1\}).
\label{eq:mpo-boundary-convention}
\end{equation}
In particular \(\beta_L=0\), \(\alpha_M=0\), and \(\alpha_L=0\).

For \(0\leq m\leq M\), let \(\mathcal I_m\) be the independent sets of the
interval family
\begin{equation}
B_j\leftrightarrow [j,j+1],
\qquad
A_j\leftrightarrow [j,j+2],
\qquad
1\leq j\leq m,
\qquad
1\leq j\leq m-1
\label{eq:mpo-intervals}
\end{equation}
on the prefix with \(m\) short generators.  Thus the corresponding physical
prefix has \(m+1\) sites.  For \(S\in\mathcal I_m\), let \(g(S)\) be the
ordered product of the selected generators, ordered by increasing left
endpoint.  The order is immaterial inside an independent set, because the
selected generators commute.  We define
\begin{equation}
T_m(u)=\sum_{S\in\mathcal I_m}(-u)^{|S|}g(S),
\qquad
0\leq m\leq M.
\label{eq:mpo-transfer-def}
\end{equation}
The full transfer matrix of the chain is \(T_M(u)\), and
\begin{equation}
-\partial_uT_M(0)=H.
\end{equation}

\subsection{Bond-dimension-\(3\) MPO for the transfer matrix}

Introduce a three-dimensional auxiliary space with basis
\(|0\rangle,|1\rangle,|2\rangle\), and let \(E_{rs}\) be its matrix units.
The local Lax operator is
\begin{equation}
\mathcal L_j(u)=
E_{00}\otimes\I_j
-u\beta_jE_{01}\otimes Y_j
-u\alpha_jE_{02}\otimes Z_j
+E_{10}\otimes X_j
+E_{21}\otimes X_j .
\label{eq:mpo-lax-def}
\end{equation}
Equivalently,
\begin{equation}
\mathcal L_j(u)=
\begin{pmatrix}
\I_j & -u\beta_jY_j & -u\alpha_jZ_j\\
X_j  & 0             & 0\\
0    & X_j           & 0
\end{pmatrix}.
\label{eq:mpo-lax-matrix}
\end{equation}
Then
\begin{equation}
T_M(u)=
\langle 0|\mathcal L_1(u)\mathcal L_2(u)\cdots \mathcal L_L(u)|0\rangle .
\label{eq:mpo-transfer-main}
\end{equation}
The auxiliary interpretation is that state \(0\) means that no \(X\)-operator
is pending, state \(1\) records one pending \(X\), and state \(2\) records two
pending \(X\)'s.  The transition
\[
0\xrightarrow{-u\beta_jY_j}1\xrightarrow{X_{j+1}}0
\]
produces the tile \(-uB_j\), while
\[
0\xrightarrow{-u\alpha_jZ_j}2\xrightarrow{X_{j+1}}1
  \xrightarrow{X_{j+2}}0
\]
produces the tile \(-uA_j\).  Thus the MPO is an automaton for non-overlapping
tiles of lengths \(2\) and \(3\).

Setting $\beta_j=0$ in \eqref{eq:mpo-lax-matrix} we obtain the Lax operators that were found by Fendley in \cite{ffd}. 

\begin{proposition}
The MPO formula \eqref{eq:mpo-transfer-main} reproduces the independent-set
transfer matrix \eqref{eq:mpo-transfer-def}.
\end{proposition}

\begin{proof}
For \(0\leq n\leq L\), define the partial contractions
\begin{equation}
F_n^{(r)}(u)=
\langle0|\mathcal L_1(u)\cdots \mathcal L_n(u)|r\rangle,
\qquad r=0,1,2,
\label{eq:mpo-partial-contractions}
\end{equation}
with \(F_0^{(0)}=\I\) and \(F_0^{(1)}=F_0^{(2)}=0\).  From
\eqref{eq:mpo-lax-matrix},
\begin{align}
F_n^{(0)}
&=
F_{n-1}^{(0)}+F_{n-1}^{(1)}X_n,
\label{eq:mpo-F0-step}
\\
F_n^{(1)}
&=
-u\beta_nF_{n-1}^{(0)}Y_n+F_{n-1}^{(2)}X_n,
\label{eq:mpo-F1-step}
\\
F_n^{(2)}
&=
-u\alpha_nF_{n-1}^{(0)}Z_n.
\label{eq:mpo-F2-step}
\end{align}
Eliminating \(F^{(1)}\) and \(F^{(2)}\) gives, for \(n\geq 3\),
\begin{equation}
F_n^{(0)}
=
F_{n-1}^{(0)}
-uB_{n-1}F_{n-2}^{(0)}
-uA_{n-2}F_{n-3}^{(0)} .
\label{eq:mpo-F0-closed-recursion}
\end{equation}
The initial values are
\begin{equation}
F_0^{(0)}=\I,
\qquad
F_1^{(0)}=\I,
\qquad
F_2^{(0)}=\I-uB_1 .
\end{equation}
On the other hand, the independent-set transfer matrices satisfy
\begin{equation}
T_m(u)=T_{m-1}(u)-uB_mT_{m-2}(u)-uA_{m-1}T_{m-3}(u),
\label{eq:mpo-independent-recursion}
\end{equation}
with \(T_{-1}(u)=T_0(u)=\I\) and with \(A_0=0\).  This follows by splitting an
independent set according to the right boundary: it either contains no
generator touching the right endpoint, or it contains \(B_m\), or it contains
\(A_{m-1}\).  Comparing \eqref{eq:mpo-F0-closed-recursion} with
\eqref{eq:mpo-independent-recursion} gives
\[
F_n^{(0)}(u)=T_{n-1}(u).
\]
Setting \(n=L=M+1\) proves \eqref{eq:mpo-transfer-main}.
\end{proof}

\subsection{RTT relation}

The local Lax operator admits a site-independent auxiliary \(R\)-matrix.  Let
\begin{align}
\mathcal R(u,v)
={}&
(u+v)\sum_{i=0}^{2}E_{ii}\otimes E_{ii}
+(v-u)\sum_{0\leq i<j\leq 2}E_{ii}\otimes E_{jj}
+(u-v)\sum_{0\leq i<j\leq 2}E_{jj}\otimes E_{ii}
\nonumber\\
&\hspace{2em}
+2u\sum_{0\leq i<j\leq 2}E_{ij}\otimes E_{ji}
+2v\sum_{0\leq i<j\leq 2}E_{ji}\otimes E_{ij}.
\label{eq:mpo-R-matrix}
\end{align}
Then we find the following remarkable identity:
\begin{equation}
\mathcal R_{12}(u,v)\,
\mathcal L_{1j}(u)\,
\mathcal L_{2j}(v)
=
\mathcal L_{2j}(v)\,
\mathcal L_{1j}(u)\,
\mathcal R_{12}(u,v).
\label{eq:mpo-local-RTT}
\end{equation}
The subscripts \(1,2\) refer to the two auxiliary spaces.

We stress that $R_{12}(u,v)$ does not depend on the inhomogeneities, but the Lax operators explicitly depend on
$\beta_j$ and $\alpha_j$. In the special case $\beta_j=0$ we obtain the RLL relation derived by Fendley in
\cite{ffd}. Thus, from the integrability point of view, the whole model family belongs to the same integrable
structure, and only the Lax operator changes as we switch on the perturbation. 

The verification of the RLL relation uses only the Pauli identities
\begin{equation}
X_jY_j=\ii Z_j,\qquad
Y_jX_j=-\ii Z_j,\qquad
Z_jX_j=\ii Y_j,\qquad
X_jZ_j=-\ii Y_j,
\label{eq:mpo-pauli-products}
\end{equation}
together with \(X_j^2=Y_j^2=Z_j^2=\I_j\).  After substituting
\eqref{eq:mpo-lax-matrix}, both sides of \eqref{eq:mpo-local-RTT} are
\(9\times9\) auxiliary matrices whose entries are linear combinations of
\(\I_j,X_j,Y_j,Z_j\); comparing these four coefficients entry by entry gives
\eqref{eq:mpo-local-RTT}.

Multiplying \eqref{eq:mpo-local-RTT} over \(j=1,\dots,L\) gives the global RTT
relation for
\[
\Omega(u)=\mathcal L_1(u)\mathcal L_2(u)\cdots\mathcal L_L(u).
\]
Since
\[
\mathcal R(u,v)|00\rangle=(u+v)|00\rangle,
\qquad
\langle00|\mathcal R(u,v)=(u+v)\langle00|,
\]
sandwiching the global RTT relation between \(\langle00|\) and \(|00\rangle\)
gives
\begin{equation}
[T_M(u),T_M(v)]=0.
\label{eq:mpo-transfer-commutativity}
\end{equation}
For \(u+v=0\) the same identity follows by polynomial continuation.

\subsection{The dressed edge operator}

The right boundary Pauli operator
\begin{equation}
\chi_R=Z_L
\label{eq:mpo-right-edge}
\end{equation}
anti-commutes precisely with \(B_M\) and \(A_{M-1}\), and commutes with all
other \(B_j,A_j\).  The dressed edge operator is
\begin{equation}
\Psi_L(u)=T_M(-u)\,Z_L\,T_M(u).
\label{eq:mpo-dressed-edge}
\end{equation}
At a zero \(u=u_k\) of the scalar polynomial associated with
\(T_M(-u)T_M(u)\), the operators \(\Psi_L(\pm u_k)\) give the corresponding
fermionic eigenmode operators after normalization.

For later use, set
\begin{equation}
x_j=u\beta_j,
\qquad
y_j=u\alpha_j,
\qquad
1\leq j\leq M,
\label{eq:mpo-x-y-def}
\end{equation}
with \(y_M=0\).  Define the scalar prefix polynomial \(p_m(u)\) by
\begin{equation}
p_m(u)=
\sum_{S\in\mathcal I_m}
(-u^2)^{|S|}
\prod_{g\in S}w(g)^2,
\qquad
w(B_j)=\beta_j,
\qquad
w(A_j)=\alpha_j.
\label{eq:mpo-scalar-polynomial-def}
\end{equation}
Equivalently,
\begin{equation}
p_{-1}(u)=p_0(u)=1,
\qquad
p_1(u)=1-x_1^2,
\label{eq:mpo-p-initial}
\end{equation}
and for \(m\geq2\),
\begin{equation}
p_m(u)=p_{m-1}(u)-x_m^2p_{m-2}(u)-y_{m-1}^2p_{m-3}(u).
\label{eq:mpo-p-recursion}
\end{equation}
Thus \(p_M(u)\) is the full scalar polynomial, written as a polynomial in
\(u^2\).

\subsection{Naive doubled MPO and the final-site boundary}

Starting from \eqref{eq:mpo-transfer-main}, the naive doubled MPO for
\(\Psi_L(u)\) has auxiliary dimension \(9\).  For \(1\leq j\leq M\), define
\begin{equation}
\mathcal K_j(u)=\mathcal L_j(-u)\otimes \mathcal L_j(u).
\label{eq:mpo-K-def}
\end{equation}
The tensor product is over the two auxiliary spaces, while the physical
entries are multiplied in the order in which they appear in
\(T_M(-u)Z_LT_M(u)\).

The last physical site is more naturally kept as a right boundary vector,
because \(Z_L\) sits between the two final Lax factors.  In the standard basis
\(|a\rangle\otimes |b\rangle\), \(0\leq a,b\leq2\), define
\begin{equation}
\bigl(|r_9(u)\rangle\bigr)_{ab}
=
\bigl(\mathcal L_L(-u)\bigr)_{a0}
\,Z_L\,
\bigl(\mathcal L_L(u)\bigr)_{b0}.
\label{eq:mpo-r9-def}
\end{equation}
Since \(\beta_L=\alpha_L=0\), the only nonzero entries are
\begin{equation}
(r_9)_{00}=Z_L,
\qquad
(r_9)_{11}=-Z_L,
\qquad
(r_9)_{01}=\ii Y_L,
\qquad
(r_9)_{10}=-\ii Y_L.
\label{eq:mpo-r9-entries}
\end{equation}
Then
\begin{equation}
\Psi_L(u)
=
\langle00|
\mathcal K_1(u)\mathcal K_2(u)\cdots\mathcal K_M(u)
|r_9(u)\rangle .
\label{eq:mpo-psi-naive}
\end{equation}

\subsection{Exact compression from bond dimension \(9\) to \(6\)}

Let \(\mathcal W_6\subset \mathbb C^3\otimes\mathbb C^3\) be spanned by
\begin{align}
s_1&=|00\rangle,
&
s_2&=|11\rangle,
&
s_3&=|22\rangle,
\nonumber\\
s_4&=\frac{|20\rangle-|02\rangle}{2},
&
s_5&=\frac{|10\rangle-|01\rangle}{2},
&
s_6&=\frac{|21\rangle-|12\rangle}{2}.
\label{eq:mpo-W6-basis}
\end{align}
The factors \(1/2\) are a basis convention chosen so that the reduced tensor
has no extra factors \(1/2\).  With row-vector multiplication, a direct local
calculation gives
\begin{align}
s_1\mathcal K_j
&=
s_1
-x_j^2s_2
-y_j^2s_3
+2y_jZ_js_4
+2x_jY_js_5
+2\ii x_jy_jX_js_6,
\nonumber\\
s_2\mathcal K_j
&=
s_1,
\nonumber\\
s_3\mathcal K_j
&=
s_2,
\nonumber\\
s_4\mathcal K_j
&=
X_js_5-\ii y_jY_js_6,
\nonumber\\
s_5\mathcal K_j
&=
-\ii y_jY_js_4+\ii x_jZ_js_5,
\nonumber\\
s_6\mathcal K_j
&=
s_5.
\label{eq:mpo-W6-action}
\end{align}
Hence \(\mathcal W_6\) is invariant under every \(\mathcal K_j(u)\).  Moreover,
\(|00\rangle=s_1\) and the boundary vector \(|r_9(u)\rangle\) belongs to
\(\mathcal W_6\).  In the basis \eqref{eq:mpo-W6-basis}, the reduced local
tensor is
\begin{equation}
\mathsf M_j^{(6)}(u)=
\begin{pmatrix}
\I_j & -x_j^2\I_j & -y_j^2\I_j
     & 2y_jZ_j & 2x_jY_j & 2\ii x_jy_jX_j\\
\I_j & 0 & 0 & 0 & 0 & 0\\
0 & \I_j & 0 & 0 & 0 & 0\\
0 & 0 & 0 & 0 & X_j & -\ii y_jY_j\\
0 & 0 & 0 & -\ii y_jY_j & \ii x_jZ_j & 0\\
0 & 0 & 0 & 0 & \I_j & 0
\end{pmatrix}.
\label{eq:mpo-M6}
\end{equation}
The reduced right boundary is
\begin{equation}
|r_6(u)\rangle=
\begin{pmatrix}
Z_L\\
-Z_L\\
0\\
0\\
-\ii Y_L\\
0
\end{pmatrix},
\qquad
\langle \ell_6|=(1,0,0,0,0,0).
\label{eq:mpo-r6}
\end{equation}
Therefore
\begin{equation}
\Psi_L(u)=
\langle \ell_6|
\mathsf M_1^{(6)}(u)\mathsf M_2^{(6)}(u)\cdots
\mathsf M_M^{(6)}(u)
|r_6(u)\rangle .
\label{eq:mpo-psi-M6}
\end{equation}

\subsection{Exact compression from bond dimension \(6\) to \(4\)}

The tensor \(\mathsf M_j^{(6)}(u)\) is block upper triangular.  Its first three
states form a scalar block,
\begin{equation}
S_j(u)=
\begin{pmatrix}
1 & -x_j^2 & -y_j^2\\
1 & 0      & 0\\
0 & 1      & 0
\end{pmatrix}.
\label{eq:mpo-scalar-block}
\end{equation}
Starting from the scalar row vector \((1,0,0)\), the first component after
\(r\) sites is \(p_{r-1}(u)\).  In particular, before site \(j\) the first
scalar component is \(p_{j-2}(u)\).  Since the coupling from the scalar block
to the non-scalar block occurs only from the first scalar state, we can
integrate out the scalar block exactly.

The resulting bond-dimension-\(4\) tensor acts on the basis consisting of one
integrated scalar state, followed by \(s_4,s_5,s_6\).  In the following display
\(p_{j-2}\) means \(p_{j-2}(u)\), with \(p_{-1}(u)=p_0(u)=1\):
\begin{equation}
\mathsf M_j^{(4)}(u)=
\begin{pmatrix}
\I_j
&
2y_jp_{j-2}Z_j
&
2x_jp_{j-2}Y_j
&
2\ii x_jy_jp_{j-2}X_j
\\
0 & 0 & X_j & -\ii y_jY_j\\
0 & -\ii y_jY_j & \ii x_jZ_j & 0\\
0 & 0 & \I_j & 0
\end{pmatrix}.
\label{eq:mpo-M4}
\end{equation}
It remains to compute the scalar part of the final boundary.  After \(M\)
sites, the scalar row has first component \(p_{M-1}(u)\).  If one appended the
last zero-weight site \(L=M+1\), the new first component would be \(p_M(u)\).
Thus the second scalar component before the last site is
\(p_M(u)-p_{M-1}(u)\).  Contracting with the scalar part
\((Z_L,-Z_L,0)^t\) of \eqref{eq:mpo-r6} gives
\[
p_{M-1}Z_L-\bigl(p_M-p_{M-1}\bigr)Z_L
=
\bigl(2p_{M-1}-p_M\bigr)Z_L.
\]
Therefore
\begin{equation}
|r_4(u)\rangle=
\begin{pmatrix}
\bigl(2p_{M-1}(u)-p_M(u)\bigr)Z_L\\
0\\
-\ii Y_L\\
0
\end{pmatrix},
\qquad
\langle \ell_4|=(1,0,0,0).
\label{eq:mpo-r4}
\end{equation}
The final compressed formula is
\begin{equation}
\Psi_L(u)=
\langle \ell_4|
\mathsf M_1^{(4)}(u)\mathsf M_2^{(4)}(u)\cdots
\mathsf M_M^{(4)}(u)
|r_4(u)\rangle .
\label{eq:mpo-psi-M4}
\end{equation}
The compressions
\[
9\longrightarrow 6\longrightarrow 4
\]
are exact for arbitrary inhomogeneous coefficients \(\alpha_j,\beta_j\).

To our knowledge, this MPO construction is new even for the original FFD model of Fendley.


\providecommand{\href}[2]{#2}\begingroup\raggedright\endgroup

\end{document}